\newif\iftwocol\twocolfalse     

\documentclass[12pt,twoside]{article}
 \newdimen\paravsp  \paravsp=1.3ex 
\usepackage{graphicx} 
\usepackage{latexsym,amsmath,amssymb} 
\usepackage{amsthm} 
\usepackage{hyperref} 
\usepackage{thmtools} 
\usepackage[capitalize]{cleveref} 
\crefname{equation}{}{} \Crefname{equation}{}{} 
\crefname{figure}{Figure}{Figures} 
\crefname{open}{Open Problem}{Open Problems} 
\usepackage{color}

\def\,{\mskip 3mu} \def\>{\mskip 4mu plus 2mu minus 4mu} \def\;{\mskip 5mu plus 5mu} \def\!{\mskip-3mu}
\newtheoremstyle{myTheoremStyle}{\topsep}{\topsep}{\itshape}{0pt}{\bfseries\boldmath}{}{5pt}{\thmname{#1}\thmnumber{ #2}\thmnote{ (#3)}}
\theoremstyle{myTheoremStyle}
\newtheorem{theorem}{Theorem}
\newtheorem{corollary}[theorem]{Corollary}
\newtheorem{lemma}[theorem]{Lemma}
\newtheorem{definition}[theorem]{Definition}
\newtheorem{proposition}[theorem]{Proposition}
\newtheorem{assumption}[theorem]{Assumption}
\newenvironment{keywords}{\centerline{\bf\small
Keywords}\begin{quote}\small}{\par\end{quote}\vskip 1ex}
\newtheorem{open}[theorem]{Open Problem}
\newtheorem{remark}[theorem]{Remark}
\renewenvironment{proof}{{\vskip 1ex\noindent\bf Proof.}}{\qed\vskip 1ex}
\iftwocol\def\tc#1{#1}\def\ntc#1{}\else\def\tc#1{}\def\ntc#1{#1}\fi 

\def\paradot#1{\vspace{\paravsp plus 0.5\paravsp minus 0.5\paravsp}\noindent{\bf\boldmath{#1.}}} 
\def\hrefurl#1{\href{#1}{\rule{0ex}{1.7ex}\color{blue}\underline{\smash{#1}}}} 

\def\eps{\varepsilon}           
\def\epstr{\epsilon}            
\def\nq{\hspace{-1em}}          
\def\qed{\hspace*{\fill}\rule{1.4ex}{1.4ex}$\quad$\\} 
\def\eoe{\hspace*{\fill} $\blacklozenge\quad$} 
\def\eor{\hspace*{\fill} {\rm\LARGE\textbullet}$\quad$} 
\AtEndEnvironment{example}{\eoe\endexample}
\AtEndEnvironment{remark}{\eor\endremark}
\def\fr#1#2{{\textstyle{#1\over#2}}} 
\def\SetR{\mathbb{R}}           
\def\SetN{\mathbb{N}}           
\def\SetB{\mathbb{B}}           
\def\SetZ{\mathbb{Z}}           
\def\E{{\mathbb E}}             
\def\lb{{\log}}                 

\def\cA{{\cal A}}                
\def\cC{{\cal C}}                
\def\cI{{\cal I}}                
\def\cU{{\cal U}}                
\def\cV{{\cal V}}                
\def\cW{{\cal W}}                
\def\cX{{\cal X}}                
\def\cY{{\cal Y}}                

\def\Km{{K\!m}}

\def\BsdKs{\SetB_K^*}             
\def\BsdKx{(\SetB^*,d_K)}         

\iftrue 
\makeatletter
\@ifundefined{@tempboxa}{\@nameuse{newbox}\@tempboxa}{}
\@ifundefined{@tempboxb}{\@nameuse{newbox}\@tempboxb}{}
\def\xxstackrel#1#2#3#4#5#6{\setbox\@tempboxa=\hbox{$#3$}
                      \mathrel{%
                        \setbox\@tempboxb=\hbox{%
                        \m@th \hbox {\ooalign {%
                                 \raise #2\hbox to \wd\@tempboxa%
                                    {\hfil$\scriptstyle #1$\hfil}\crcr%
                                 \lower #4\box\@tempboxa}}}%
                        \ht\@tempboxb=#5
                        \dp\@tempboxb=#6
                        \box\@tempboxb\relax
                          }}
\makeatother
\def\equa{\xxstackrel{+}{1.2ex}{=}{0ex}{1.6ex}{0.45ex}}
\def\leqa{\xxstackrel{+}{1.4ex}{\leq}{0.1ex}{2ex}{0.45ex}}
\def\geqa{\xxstackrel{+}{1.4ex}{\geq}{0.1ex}{2ex}{0.45ex}}
\def\equs{\xxstackrel{s}{1.2ex}{=}{0ex}{1.6ex}{0.45ex}}
\def\lesstar{\xxstackrel{*}{1.2ex}{<}{0ex}{1.6ex}{0.45ex}}
\def\leqs{\xxstackrel{s}{1.4ex}{\leq}{0.1ex}{2ex}{0.45ex}}
\def\geqs{\xxstackrel{s}{1.4ex}{\geq}{0.1ex}{2ex}{0.45ex}}
\else 
\def\equa{\smash{\;\stackrel+=\;}}
\def\leqa{\smash{\;\stackrel+\leq\;}}
\def\geqa{\smash{\;\stackrel+\geq\;}}
\def\equs{\smash{\;\stackrel{s}=\;}}
\def\lesstar{\smash{\;\stackrel*<\;}}
\def\leqs{\smash{\;\stackrel{s}\leq\;}}
\def\geqs{\smash{\;\stackrel{s}\geq\;}}
\fi

\begin{document}


\title{\vspace{-4ex}
\vskip 2mm\bf\Large\hrule height5pt \vskip 4mm
Properties of Algorithmic Information Distance
\vskip 4mm \hrule height2pt}
\author{{\bf Marcus Hutter}\\[3mm]
\normalsize Google DeepMind \& Australian National University\\[2mm]
\normalsize \hrefurl{http://www.hutter1.net/publ/kmetric.pdf}
}
\date{July 2025}
\maketitle

\begin{abstract}
The domain-independent universal Normalized Information Distance based on Kolmogorov complexity
has been (in approximate form) successfully applied to a variety of difficult clustering problems.
In this paper we investigate theoretical properties of the un-normalized algorithmic information distance $d_K$.
The main question we are asking in this work is what properties this curious distance has, besides being a metric.
We show that many (in)finite-dimensional spaces can(not) be isometrically 
scale-embedded into the space of finite strings with metric $d_K$.
We also show that $d_K$ is not an Euclidean distance, 
but any finite set of points in Euclidean space can be scale-embedded into $(\{0,1\}^*,d_K)$.
A major contribution is the development of the necessary framework 
and tools for finding more (interesting) properties of $d_K$ in future, 
and to state several open problems.
\vspace{5ex}\def\contentsname{\centering\normalsize Contents}\setcounter{tocdepth}{1}
{\parskip=-2.7ex\tableofcontents}
\end{abstract}


\section{Introduction}\label{sec:Intro}

\paradot{Metric spaces and embeddings}
Isomorphisms and isomorphic embeddings play an important role in all parts of mathematics:
to relate different mathematical structures or even mathematical sub-fields, 
to reduce problems to each other,
to deepen our understanding, and of course for mathematical `economy'.
Isometries are isomorphisms between metric spaces that preserve distance.
An isometric embedding $\phi:\cV\to\cW$ from a metric space $\cV$ with distance $d_\cV$
to a metric space $\cW$ with metric $d_\cW$ by definition 
satisfies $d_\cW(\phi(v),\phi(v'))=d_\cV(v,v')$.
For instance, orthogonal transformations from $\cV=\SetR^m$ to $\cW=\SetR^m$
preserve the Euclidean distance $d_\cV=d_\cW=||\cdot||_2$. 
There also exist universal metric spaces $(\cU,d_\cU)$,
meaning \emph{all} metric spaces can be isometrically embedded into $(\cU,d_\cU)$.
Studying the relation between metric spaces, 
and their additional properties has a long history \cite{Frechet:10,Urysohn:27,Husek:08,Graham:85}
and has even been featured in recreational/artistic work \cite{Zeng:20}.

\paradot{Euclidean spaces}
Euclidean spaces such as $(\SetR^m,||\cdot||_2)$ for $m<\infty$ and 
their infinite-dimensional extension, 
Hilbert spaces, for $m=\infty$ (also called infinite-dimensional Euclidean spaces) 
play a particularly important role.
Not all metric spaces are Euclidean.
So whether a metric has the Euclidean property, that is,
whether it can be isometrically embedded into $(\SetR^m,||\cdot||_2)$ or a Hilbert space,
is one of the most natural and first questions to ask,
since a lot is known about Euclidean spaces,
and Euclidean spaces enjoy particularly nice properties.
For instance, unlike most other interesting metrics one can impose on $\SetR^m$,
e.g.\ $p$-norms, the Euclidean distance on $\SetR^m$ 
has a very large automorphism group (isomorphisms onto itself).
An important application is in machine learning:
If $d(x,y)$ measures some (abstract) distance between data points,
then $k_\lambda(x,y):=\exp(-\lambda d(x,y)^2)$ 
is a kernel which measures the similarity between $x$ and $y$.
$k_\lambda$ is a positive definite kernel ($\forall\lambda>0$)
if and only if $d$ can be isometrically embedded into some Euclidean 
space (the kernel trick), which allows one to 
apply linear methods for classification (such as SVMs) 
or (e.g.\ logistic) regression in this Euclidean feature space.
Many algorithms have been developed to approximately map points from
a high-dimensional or abstract metric space into a lower-dimensional Euclidean space
\cite[Chp.15]{Matousek:02}.
Not least, for $m\leq 3$ this is our local space we live in, 
so objects in $\SetR^{\leq 3}$ have convenient representations
(see e.g.\ the popular $t$-SNE \cite{VanderMaaten:08} algorithm).

\paradot{Algorithmic information distance}
The algorithmic information distance \cite{Bennett:98} 
is a peculiar complexity metric $d_K$ 
over the set of finite binary string $\SetB^*$.
It measures the distance of (an object encoded as) string $x$ 
from (an object encoded as) string $y$ in terms of the length 
of the shortest computer program that transforms $x$ into $y$ and $y$ into $x$.
If $x$ and $y$ share a lot of structure (in the extreme case $x=y$),
then $d_K(x,y)$ is small. Otherwise it will be large(r), 
e.g.\ if $x$ is a piece of text and $y$ is a digitized photo.
But if $y$ is a photo of the text $x$, $d_K$ would be smaller,
since $y$ can be compressed with the help of $x$ 
(and $x$ is close to trivial given $y$, think OCR).
Most mathematically ``natural'' distances would not capture this relation.
The distance is universal in the sense that it minorizes most other distances 
(like Hamming or Edit distance).
By normalizing it (NID) \cite{Li:04} and approximating it 
by standard file compressors like gzip (NCD) \cite{Cilibrasi:05},
it has led to impressive applications including a completely automatic reconstruction 
of the evolutionary tree of 24 mammals based on complete mtDNA, 
of the classification tree of 52 languages based on the declaration of human rights,
of the phylogeny of SARS-CoV-2 virus \cite{Cilibrasi:22},
and various others \cite{Cilibrasi:05,Jiang:23fewshot,Jiang:23text},
without using any domain-specific knowledge.
Its theoretical properties have barely been studied.
It has been shown that it is indeed (nearly) a metric,
and that it minorizes all other metrics with limited neighborhood size \cite[Sec.III]{Li:04}.
Very recently $d_K$ has been applied in kernel methods \cite{Hutter:25kccluster}.
For instance, while we show that $d_K$ is not Euclidean on all of $\SetB^*$,
it can accommodate any finite Euclidean point set, 
so the kernel trick (see above) could be applied to kernel $\exp(-\lambda d_K^2)$ on such subspace.

\paradot{This paper}
The goal of this paper is to initiate a deeper study of
the un-normalized complexity distance $d_K$.
We will encounter sufficient obstacles even for this somewhat simpler cousin of NID.
We expect many results to transfer in some form or another to NID,
but this is beyond the scope of this article.
It is known that $d_K$ is also a metric and has a certain minorization property.

Besides the properties of $d_K$ we were able to prove,
the main contribution of the paper is to provide the necessary framework and tools 
for finding more (interesting) properties of $d_K$ in future, 
and to state some open problems, which hopefully stirs interest.

One way to study the properties of metric space $\BsdKs:=\BsdKx$,
is to find isometric embeddings from or into other metric spaces or into itself.
This is the approach we take in this paper.
By showing that $\BsdKs$ isometrically embeds into some $(\cV^*,d_\cV)$,
it will inherit any ``positive'' properties of the latter.
For instance, if we could isometrically embed $\BsdKs$ into a Hilbert space,
we would have shown that $d_K$ is Euclidean.
Conversely, if we isometrically embed some $(\cV^*,d_\cV)$ into $\BsdKs$,
the latter would inherit any ``negative'' property of the former.
Indeed, we will embed a non-Euclidean space into $\BsdKs$,
showing that $d_K$ is \emph{not} Euclidean. On the other hand,
this does not preclude that $d_K$ is Euclidean on some sub-spaces of $\SetB^*$.
Possibly $\BsdKs$ is a universal $\langle${\it qualifier}$\rangle$ metric space,
where $\langle${\it qualifier}$\rangle$ has yet to be determined,
and could include/embed rich Euclidean and non-Euclidean spaces.

\paradot{Contents}
\Cref{sec:Metrics} introduces notation,
(further) properties distance functions may possess,
and in particular (Euclidean) metrics.
A minor variation of the algorithmic information distance 
is a proper metric, which we will call `complexity metric'.
\Cref{sec:IsoEmbedding} briefly discusses universal metric spaces,
and isometric embeddings into $\ell_1$ and the Hamming cube.
Due to the discrete nature and additive constants in $d_K$,
we need to define the relaxed notion of `scale-isometries',
which allows to embed the Hamming cube 
and bounded subspaces of $(\SetR^m,||\cdot||_1$) into $\BsdKs$.
In \cref{sec:RandomEmbedding} we prove a general theorem that allows to find scale-embeddings 
via probability measures and string-valued random variables.
We use this to prove a (restricted) product space embedding theorem.
\Cref{sec:KraftIneq} looks for extra properties which $d_K$ possesses
that may be useful to prove negative embedding results. 
For instance, Kraft's inequality implies that any $d_K$-ball 
of radius $r$ contains at most $2^r$ points.
This prevents embedding many infinite-dimensional spaces into $\BsdKs$.
We also show that not every finite metric space embeds into $\BsdKs$.
\Cref{sec:Euclid} shows that $\BsdKs$ is not Euclidean, 
but any finite Euclidean point set can be scale-embedded.
\Cref{sec:Disc} concludes.
\Cref{app:Notation} contains a complete list of notation.

Most sections identify further open problems, e.g.\ 
whether $d_K$ is universal in some interesting sense,
whether any finite metric space can be scale-embedded into $d_K$,
whether a small power of $d_K$ is Euclidean, and others.
Some more elementary calculations or proofs 
that don't require creativity are marked as `(Exercise)'.

\section{Distances and Metrics}\label{sec:Metrics}

We start by establishing basic notation,
define positive and conditionally negative definite kernels,
and various distances and metric spaces and their properties,
including Hamming distance, Euclidean metrics, 
(Shannon) information pseudo-metric, Jensen-Shannon divergence, 
and the algorithmic information distance or complexity metric 
based on prefix Kolmogorov complexity.

\paradot{Notation}
$m$ will usually denote the dimension of a matrix or vector space;
$i,j$ are the vector and matrix indices $\in\{1,...,m\}$;
$s$ is a scaling factor;
$k$ are kernels;
$d$ are distances;
$K$ will denote the Kolmogorov complexity;
$n$ is an integer used generically.
With $\SetB:=\{0,1\}$, $\SetB^*$ ($\SetB^m$) denotes set of all finite binary strings (of length $m$), 
and $\SetB^\infty$ the set of all infinite sequences.
$\BsdKs:=(\SetB^*,d_K)$ abbreviates $\SetB^*$ equipped with metric $d_K$.
$\cX^m$ are vectors of length $m\leq\infty$ over $\cX$.
So $\SetB^m$ is alternatively interpreted as vectors when convenient.
The Iverson bracket $[\![\text{\it bool}]\!]=1$ if $\text{\it bool}$ is true, and $0$ otherwise.
All logarithms are base $2$.

Throughout this work, $\cX$ (and $\cU,\cV,\cW$) will be a set (space) 
equipped with some function $f:\cX\times\cX\to\SetR$.
Depending on the properties of $f$ to be clarified later, 
these functions are called distances or metrics or inner products or kernels.
For an (ordered) finite subset of $\{x_1,...,x_m\}\subseteq\cX$,
we define the matrix $M^f$ with entries $M^f_{ij}:=f(x_i,x_j)$.

See \cref{app:Notation} for a more complete list of notation.

\begin{definition}[Kernels]\label{def:kernel}
A set $\cX$ equipped with some function $k:\cX\times\cX\to\SetR$
is called a kernel space, and $k$ a kernel.
The kernel usually possesses further properties, e.g.\ some of 
\begin{itemize}\parskip=0ex\parsep=0ex\itemsep=0ex
\item[(P)] Non-negative: $k(x,y)\geq 0$
\item[(Z)] Zero on diagonal: $k(x,x)=0$
\item[(N)] Non-zero off-diagonal: $k(x,y)\neq 0$ for $x\neq y$
\item[(S)] Symmetry: $k(x,y)=k(y,x)$
\item[(TI)] Triangle Inequality: $k(x,z)\leq k(x,y)+k(y,z)$
\item[(PD)] Positive Definite: $\forall m\in\SetN~\forall x_1,...,x_m\in\cX~\forall c_1,...,c_m\in\SetR: \sum_{i,j} c_i c_j k(x_i,x_j)\geq 0\nq$
\item[(CND)\tc{\nq}] \tc{\quad} Conditionally Negative Definite: \ntc{\\} 
    $\forall m\in\SetN~\forall x_i\in\cX~\forall c_i\in\SetR$: $\sum_{i=1}^m c_i=0\Rightarrow \sum_{i,j} c_i c_j k(x_i,x_j)\leq 0$
\item[(E)] Euclidean: $\exists\Phi:\cX\to H\!$ilbert space such that $k(x,y)=||\Phi(x)-\Phi(y)||_2$
\item[(IP)] Inner product: If $\cX$ is a vector space with inner product $k$.
\end{itemize}
The conditions are meant to hold for all $x,y,z,...\in\cX$.
\end{definition}

Most definitions require $k$ to be symmetric (S) (or Hermitian in the complex case).
If $k$ satisfies (P,Z,N,S,TI) it is called a metric.
Note that (E) implies but is not implied by (P,Z,S,TI).
Relaxing some of the conditions we arrive at pseudo- (no N), quasi- (no S), meta- (no Z), semi- (no TI), pre- (only P,Z) metrics.
We sometimes call a metric `proper' to emphasize that it has none of the defects above.
The word `distance' is used inconsistently in the literature, 
so we will only use it generically if some of the first 5 properties hold.

\paradot{Compression distance}
There are many different distances, suitable for different applications.
A natural quest is for \emph{universal} distances which subsume as many favorable 
properties as possible, and can be used as widely as possible.
One such distance is based on Kolmogorov complexity,
based on the idea of universal optimal data compression of 
(binary encodings of) any digital objects \cite[Chp.3]{Li:19}:

\begin{definition}[Conditional Kolmogorov Complexity]
Let $U$ be some fixed (optimal) reference prefix Universal Turing Machine (UTM), 
$p\in\SetB^*$ be a binary-encoded program and $x,y\in\SetB^*$ be some binary-encoded objects, then
\iftwocol
\begin{align*}
  K(x)~&:=~ \min\{\ell(p):U(p)=x\} \\
  K(x|y) ~&:=~ \min\{\ell(p):U(p,y)=x\}
\end{align*}
\else $$
    K(x)~:=~ \min\{\ell(p):U(p)=x\} ~~~\text{and}~~~ K(x|y) ~:=~ \min\{\ell(p):U(p,y)=x\}
$$ \fi
\end{definition}
That is, $K(x)$ ($K(x|y)$) is the length of the shortest program for $x$ (given $y$ as input).
Naturally, the more similar $x$ and $y$, the smaller $K(x|y)$,
where `similar' can be understood very flexibly,
as long as there is \emph{some} computable transformation of $y$ into $x$.
One can show that $K(x|y)$ minorizes most other distances (\Cref{thm:minorize}). 
Despite its incomputability, it has found applications in many areas of 
science, engineering, mathematics, statistics, and even the arts and humanities \cite{Li:07,Schmidhuber:12,Wikipedia:24kart}.

Most AIT results are plagued with additive constants 
which arise from constructing `compiler programs' in the proofs.
We use the notation $f\leqa g$ if $\sup_x f(x)-g(x) \leq c<\infty$, and similarly $\geqa$ and $\equa$. 
for some universal constant $c=O(1)$ independent of all parameters.
We will use the following basic properties of $K$ 
\cite[Thm.2.10]{Hutter:04uaibook}, \cite[Sec.2.7]{Hutter:24uaibook2} \cite[Chp.3]{Li:19}:

\begin{lemma}[Properties of $K$]\hfil\par\ntc{\vspace{-2ex}}\label{lem:Kprop}
\begin{itemize}\parskip=0ex\parsep=0ex\itemsep=0ex
\item[$(i)$] $K(x) ~\leqa~ \ell(x)+1.1\log\ell(x) ~\leqa 1.1\ell(x) ~\leq~ 2\ell(x)$ ~~~\text{and}~~~\tc{\\} $K(n)\leqa 1.1\lb n$
\item[$(ii)$] $K(x|y) ~\leqa~ K(x) ~\leqa~ K(x,y) ~\tc{\\}\leqa~ K(x|y)+K(y) ~\leq~ K(x)+K(y)$
\item[$(iii)$] $K(x) ~\leqa~ -\log P(x) + K(P)~$ if $~\sum_{x\in\SetB^*} P(x)\leq 1$ ~~~\tc{\\}(MDL bound) 
\item[$(iv)$] $K(X)+\log P(X) ~\geq~ \log\delta$ w.p.$\geq 1-\delta~$ if $~X\sim P$   ~~\tc{\\}(randomness deficiency)
\item[$(v)$] $\sum_{x\in\SetB^*}2^{-K(x|y)}<1$    \ntc{~~~~~~~~~~~~~~~~~~~~~~~~~~~~~~~~~~~}~~~ (Kraft inequality)
\end{itemize}
\end{lemma}

All bounds remain valid if conditioned on some further string $z$.
The bounds and terminology will be discussed and motivated in more detail where used.
A string $x$ is called algorithmically random ($K$-random) iff $K(x)\geqa \ell(x)$,
i.e.\ if $x$ is incompressible.
Strings $x$ and $y$ are called algorithmically independent ($K$-independent) 
iff $K(x,y)=K(x)+K(y)\pm O(\log\ell(x)+\log\ell(y)))$, 
i.e.\ if $x$ and $y$ do not share information.

Back to metrics, if we allow additive slack in \cref{def:kernel}, then $K(x|y)$ satisfies (P,Z,N,TI). 
(S) can easily be achieved by taking the maximum or average of $K(x|y)$ and $K(y|x)$
(see proof of \cref{thm:dKmetric}).
This leads to the well-known algorithmic information distance \cite{Bennett:98}.
One can also get rid of the additive slack and arrive at a proper metric,
by adding adding a constant of $O(1)$ if $x\neq y$ and setting $K(x|x)=0$.
Since $K(x|x)=O(1)$, this is again only an $O(1)$ correction (\cref{thm:dKmetric}).

The {\bf main question} we are asking in this paper is what additional properties does this metric satisfy.
On the one hand, being uniquely defined 
(apart from the choice of UTM, which for most results is irrelevant), 
it has of course very specific properties. 
The question is whether it possesses some natural properties that also other distances have.
On the other hand, its universality gives hope that it allows to embed many other distances,
which would mean it has few further properties, at least globally.

Below we present some well-known distances and metrics,
formally defined only when needed.

\begin{definition}[Distances \& (Euclidean) Metrics]\hfil\par\ntc{\vspace{-2ex}}\label{def:distance}
\begin{itemize}\parskip=0ex\parsep=0ex\itemsep=0ex
\item Discrete metric $d_{01}(x,y):=[\![x\neq y]\!]$ is Euclidean.
\item Euclidean metric $d_E:=||x-y||_2$ for $\cX=\SetR^m$ or $\cX=H\!$ilbert space.
\item Hamming metric $d_H:=||x-y||_1=||x-y||_2^2$ for $\cX=\SetB^m$, hence $\sqrt{d_H}$ is (E).$\nq$
\item Graph metric $d_G$: Length of shortest path between two nodes.
\item Information pseudo-metric $d_I^+:=\fr12[H(X|Y)+H(Y|X)]$ \\ 
      and $d_I^\vee:=\max\{H(X|Y),H(Y|X)\}$, where $H$ is the conditional entropy.
\item Jensen--Shannon divergence $d_{JS}$: Note that $\sqrt{d_{JS}}$ is an Euclidean metric \cite{Fuglede:04}.
\item Complexity metric $d_K(x,y):=\max\{K(x|y),K(y|x)\}+c$ for $x\neq y\in\cX=\SetB^*$ \ntc{\\}
      (any $c\geq c_U$, where $c_U>0$ is some universal constant), and $d_K:=0$ if $x=y$.
\end{itemize}
\end{definition}

Our results remain valid even for $c=0$,
but it is convenient not having to worry about the meaning of a sort-of metric with $O(1)$ slack,
not having to qualify it as a slack-metric or so,
and it mildly conveniences some proofs.
All results in this paper also remain valid if the $\max\{a,b\}$ in $d_K$ is replaced by
the average $\fr12[a+b]$. We use $d_K^\vee$ and $d_K^+$ if we need to distinguish them.
For the same convenience we defined $d_I^+$ with $\fr12$, which is non-standard.

\begin{theorem}[Complexity metric {\cite[Thm.4.2]{Bennett:98}}]\label{thm:dKmetric}
There exists a universal constant $c_U>0$ of $O(1)$ such that $d_K$ is a proper metric for all $c>c_U$.
\end{theorem}

\begin{proof}
First, $K(x|z)\leq K(x|y)+K(y|z)+O(1)$ follows immediately from the fact 
that a code for $y$ given $z$ together with a code 
for $x$ using $y$ also constitutes a code for $x$ given $z$.
In the same way, $K(z|x)\leq K(z|y)+K(y|x)+O(1)$.
Let $c_U$ be this constant.
Now taking the maximum on both sides we get for $x\neq y\neq z\neq x$,
\begin{align*}
   \tc{&} d_K^\vee(x,z) ~\ntc{&}\equiv~ \max\{K(x|z),K(z|x)\}+c \\ \nonumber
   ~&\leq~ \max\{K(x|y)+K(y|z),~K(z|y)+K(y|x)\}+c_U+c \\ \nonumber
  ~&\leq~ \max\{K(x|y),K(y|x)\}+\max\{K(y|z),K(z|y)\}+c_U+c \\ \nonumber
  ~&\equiv~ d_K^\vee(x,y)-c+d_K^\vee(y,z)-c+c_U+c ~\leq~ d_K^\vee(x,y)+d_K^\vee(y,z) ~~~\tc{\\&}\text{for}~~~ c\geq c_U
\end{align*}
and similarly for $d_K^+$.
If any two of $x,y,z$ are the same, the triangle inequality is trivially satisfied.
That is, $d_K$ satisfies the triangle inequality (TI).
Clearly $d_K$ is symmetric, $d_K(x,x)=0$, and $d_K(x,y)\geq c>0$ for $x\neq y$,
hence it satisfies (P,Z,N,S).
Together this shows that $d_K$ is a proper metric.
\end{proof}

\section{Scale-Isometric Embeddings}\label{sec:IsoEmbedding}

Property-preserving (isomorphic) embeddings are important tools in mathematics, 
e.g.\ to reduce problems in one domain to another,
especially isometric embeddings. 
Consider, for instance,
the Frechet embedding \cite[p.162]{Frechet:10}, a precursor of Kuratowski's embedding:

\begin{proposition}[Every separable metric space iso-embeds into $\ell_\infty$]\label{prop:linfty_embedding} 
\end{proposition}

$\ell_p:=\{v\in\SetR^\infty:||v||_p<\infty\}$ for some $p\leq\infty$ 
is the vector space of infinite sequences with finite $p$-norm, which induces metric $||v-v'||_p$.
More generally, with $(\SetR^m,||\cdot||)$ for $m\leq\infty$
we mean the metric space $\cV:=\{v\in\SetR^m:||v||<\infty\}$ with metric $d_\cV(v,v'):=||v-v'||$.

\begin{proof} Let $(\cX,d)$ be a separable metric space. Let $\cX':=\{x_0,x_1,x_2,...\}$ be a countable dense subset of $\cX$.
Consider embedding $\phi:\cX\to\SetR^\infty$ with $\phi(x)_i:=d(x,x_i)-d(x_0,x_i)$. 
Let $x,y\in\cX$ and $y_i\in\cX'$ a (sub)sequence converging to $y\in\cX$ for $i\to\infty$, 
which exists, since $\cX'$ is dense in $\cX$.
Then
\begin{align*}
    ||\phi(x)||_\infty ~&=~ \sup\nolimits_i|d(x,x_i)-d(x_0,x_i)| ~\tc{\\&}\leq~ d(x,x_0) ~<~ \infty  ~~~\text{by (TI)} \\
    ||\phi(x)-\phi(y)||_\infty ~&=~ \sup\nolimits_i |d(x,x_i)-d(y,x_i)| ~~~\tc{\\&}\leq~ d(x,y) ~~~~~~~~~~~~\text{by (TI)} \\
    ||\phi(x)-\phi(y)||_\infty ~&\geq~ \overline\lim_i [d(x,y_i)-d(y,y_i)] 
     ~\tc{\\&}\geq~ \overline\lim_i d(x,y_i)- \overline\lim_i d(y,y_i) \\
     ~&=~ d(x,y) - d(y,y) ~=~ d(x,y) ~~~\tc{\\&}\text{since all metrics are continuous.}~~ \\[-7ex]
\end{align*}
\end{proof}

Since all spaces of practical relevance are separable (otherwise they cannot be numerically approximated),
subspaces of $\ell_\infty$ would be all one needs from a purely metric perspective.
Unfortunately $\ell_\infty$ is not itself separable,
but the Urysohn universal space fixes this problem \cite{Husek:08}.
Every countable space is obviously separable.
Hence also $\BsdKs$ can be isometrically embedded into $\ell_\infty$.
So any property of the $\infty$-norm also holds for $d_K$,
but it seems that $\ell_\infty$ does not have a lot of interesting structure.
Anyway $\ell_\infty$ is somewhat unnatural for $d_K$:
%
$K(xy)\approx K(x)+K(y)$ if $x$ and $y$ are $K$-independent, 
hence $d_K(xy,x'y')\approx d_K(x,x')+ d_K(y,y')$ if $xx'$ is $K$-independent of $yy'$ (Exercise).
That is, $d_K$ is more similar to the additive Hamming distance $d_H$,
where equality holds exactly and always. 
Some of our proofs exploit this similarity.
The Hamming distance is the 1-norm on $\SetB^m$ for strings of length $m$.
This makes $\BsdKs$ intuitively closer in nature to $\ell_1$ than to $\ell_\infty$,
the latter ignoring all but the maximal coordinate.
Other $p$-norms for $1<p<\infty$ are also sensitive to all coordinates, but sub-additively so.
$d_K$ can also behave sub-additively if $xx'$ and $yy'$ share information (unlike $d_H$).
As it turns out, $d_K$ can neither be isometrically embedded into $\ell_1$, 
nor is it Euclidean, nor anything in-between (\Cref{thm:dHnotinl1}).

We will use some of the following elementary facts 
about embeddings into Hamming cubes $(\SetB^m,d_H)\equiv(\SetB^m,||\cdot||_1)$:

\begin{lemma}[Isometric embeddings into $(\SetB^m,d_H)$] For $m\in\SetN$\hfil\par\ntc{\vspace{-2ex}}\label{lem:EinHam} 
\begin{itemize}\parskip=0ex\parsep=0ex\itemsep=0ex
    \item[$(i)$] A tree graph $T$ with $m$ nodes can isometrically be embedded into $(\SetB^{m-1},d_H)$
    \item[$(ii)$] $(\{1,...,m\},|\cdot|_1)$ can be iso-embedded into $(\SetB^{m-1},d_H)$
    \item[$(iii)$] A polygon graph (ring) with $2m$ nodes can be iso-embedded into $(\SetB^m,d_H)$
    \item[$(iv)$] A polygon graph (ring) with $m$ nodes can be iso-embedded into $(\SetB^m,\fr12 d_H)$
    \item[$(v)$] All connected graphs with $4$ (or less) nodes iso-embed into $\SetB^4$ or $(\fr12\SetB)^5$
    \item[$(vi)$] If $(\cV,d_\cV)$ iso-embeds into $(\SetB^m,d_H)$ and $(\cW,d_\cW)$ into $(\SetB^{m'},d_H)$, \ntc{\\}
    then $(\cV\times\cW,d_\cV+d_\cW)$ embeds into $(\SetB^{m+m'},d_H)$.
    \item[$(vii)$] Any bounded subset of lattice $(\SetZ^m,||\cdot||_1)$ iso-embeds into some Hamming cube.
\end{itemize}    
\end{lemma}

\begin{proof}{\bf(sketch)}
    {\boldmath$(i)$} Proof by induction. Trivial for $m=1$. 
    Assume we have an embedding $\phi:v\in T\mapsto x\in\SetB^{m-1}$ for a tree $T$ with $m$ nodes $v_1,...,v_m$.
    Modify $\phi$ and append $0$ to all embedding strings.
    The $0$ does not change the Hamming distance among nodes in $T$.
    Since $T$ is a tree, $v_{m+1}$ will have a unique neighbor $v_\imath\in T$. 
    Add node $v_{m+1}$ to tree $T$, 
    and extend $\phi$ to map $v_{m+1}$ to $\phi(v_\imath)1$.
    Now $d_H(\phi(v_{m+1}),\phi(v))=d_H(\phi(v_\imath),\phi(v))+1$ 
    $\forall v\in T\setminus\{v_{m+1}\}$ as it should be.
    {\boldmath$(ii)$} Map $u\in\{1,...,m\}$ to $0^{m-u}1^{u-1}\in\SetB^{m-1}$. 
    {\boldmath$(iii)$} Map ring node $i$ to $0^{m-i+1}1^{i-1}$ for $1\leq i\leq m$ 
    and to $1^{2m-i+1}0^{i-m-1}$ for $m<i\leq 2m$.
    {\boldmath$(iv)$} Take every second point in embedding $(iii)$ and scale the cube by $\fr12$.
    {\boldmath$(v)$} By explicit construction of embeddings for the 6 possible 4-node connected graphs (Exercise). 
    3 of the graphs (those with triangles) require sub-dividing each edge with an extra node.
    {\boldmath$(vi)$} This is a straightforward consequence of the additivity of the Hamming distance.
    {\boldmath$(vii)$} Immediately follows from $(ii)$ and $(vi)$ and translation invariance of $||\cdot||_1$.
\end{proof}

\paradot{Scale-embeddings into $\BsdKs$}
In the following we consider the other direction, 
namely isometric embeddings \emph{into} $\BsdKs$.
Of course, since $d_K$ is a proper metric,
we can only embed proper metric spaces,
so if $d$ lacks one or more of (P,Z,N,S,TI) it cannot be embedded.
Even if $d$ is a proper metric,
finding exact embeddings will be difficult and also make little sense, 
due to the additive constants which plague Kolmogorov complexity.
We could consider quasi-isometric embeddings,
which only require $s'\cdot d()-c< d_K()< s\cdot d()+c$ for some $0<s'\leq 1\leq s<\infty$ and $0\leq c<\infty$.
We opt for a stronger notion with $s=s'\to\infty$ and fixed $c$.
So we scale metric $d$ with a large constant $s$,
then approximately embed the result
such that $d_K/s$ converges to $d$ for $s\to\infty$.

\begin{definition}[Scale-embedding accuracy]\label{def:equs}
$f(s,v,v',...)\equs g(s,v,v',...)$ :iff 
$\exists c,b: |f(s,v,v',...)-g(s,v,v',...)|\leq c\log s+b$,
where $c,b$ are independent of $s\in\SetN$ and $v,v'\in\cV$ 
but may depend on $...$ such as $\cV,d_\cV,\phi,K,...$.
\end{definition}

Most results in Kolmogorov complexity hold within additive logarithmic terms.
This motivates the $\log s$ slack. 
Ideally we want some accuracy guarantee $\eps$ for $d_K(v,v')$ to hold for \emph{all} $v,v'\in\cV$.
This motivates the independence of the constants from $v,v'\in\cV$.

\begin{definition}[Scale-embedding into $\BsdKs$]\label{def:scale_embed}
Let $(\cV,d)$ be a metric space. A scale-embedding is a mapping 
that preserves distances apart from a constant scale factor $s>0$:
We call a collection of functions $\phi_s:\cV\to\SetB^*$ a (finite) string-embedding 
if $d_K(\phi_s(v),\phi_s(v')) \equs s\cdot d(v,v')$ for all $v,v'\in\cV$ and $s\in\SetN$.
We call $\phi:\cV\to\SetB^\infty$ an (infinite) sequence-embedding
if $d_K(\phi(v)_{1:f(s,v)},\phi(v')_{1:f(s,v')}) \equs s\cdot d(v,v')$ 
for some function $f:\SetN\times\cV\to\SetN$ (w.l.g.) monotone increasing in $s$.
\end{definition}

If we define $\phi_s(v):=\phi(v)_{1:f(s,v)}$,
we see that the $\phi_s$ indeed form a scale-embedding, but with
the additional property that string $\phi_{s'}(v)$ extends string $\phi_s(v)$ for $s'>s$.
We also abbreviate $x_v^s:=\phi_s(v)$, sometimes dropping the $s$ if not important.
For sequence-embeddings, using so-called (conditional) monotone complexity $\Km$ 
\cite[Def.2.7.4]{Hutter:24uaibook2} \cite[Def.6]{Hutter:07postbndx} \cite[Sec.6.3]{Shen:17}
instead of $K$ is more natural,
but since $\Km$ equals $K$ within logarithmic additive terms,
we can also use $K$ provided $f(s,\cdot)$ grows only polynomially in $s$.
In most cases, $f$ will grow linearly with $s$, since $d_K$ needs to scale linearly with $s$,
which can be achieved by linear $f$ provided the constructed strings $\phi_s(v)$ are sufficiently random.
For simplicity we make the following mild extra global assumption.
Some results rely on this:

\begin{assumption}[Monotone complexity]\label{ass:Km}
For sequence-embeddings, we either assume $f(s,\cdot)$ grows at most polynomially in $s$,
or $K$ is replaced by monotone complexity $\Km$.
\end{assumption}

\begin{lemma}[Basic scale-embedding properties]\hfil\par\ntc{\vspace{-2ex}}\label{lem:SEmbedBasic}
\begin{itemize}\parskip=0ex\parsep=0ex\itemsep=0ex
  \item[$(i)$] $\max\{f(s,v),f(s,v')\}\geqs s\cdot d(v,v')$.
  \item[$(ii)$] For every $v\neq v'$, $\phi_s(v)\neq\phi_s(v')$ for almost all $s$.
  \item[$(iii)$] If $\cV$ is finite, then $\phi_s$ is injective for almost all $s$.
  \item[$(iv)$] The collection $\phi_*:\cX\to(\SetB^*)^\infty$ with $\phi_*=(\phi_1,\phi_2,...)$ is injective.
  \item[$(v)$] Sequence-embeddings $\phi:\cV\to\SetB^\infty$ are injective.
\end{itemize}
\end{lemma}

Scale-embeddings are probably best viewed as $\phi_*$.
For fixed $s$, $d_\cV$ can only approximately be recovered from $d_K$
to relative accuracy $|d_K(x_v^s,x_{v'}^s)/d_\cV(v,v')-1|=O((\log s)/s)$,
but we can recover $d_\cV$ exactly from $\phi_*$ via $d_\cV(v,v')=\lim_{s\to\infty} d_K(x_v^s,x_{v'}^s)/s$.
In this sense, $\phi_*$ and $\phi$ are \emph{exact} isometric embeddings,
despite the $O(\log s)$ slack.

\begin{proof}
{\boldmath$(i)$} $f(s,v) = \ell(x_v^s) \geqs K(x_v^s) \geqa K(x_v^s|x_{v'}^s)$ \tc{\\}
$~\Rightarrow~  \max\{f(s,v),f(s,v')\} \geqs d_K(x_v^s,x_{v'}^s) \equs s\cdot d(v,v')$ \\
{\boldmath$(ii)$} For $v\neq v'$, $\exists s_{v,v'}\forall s>s_{v,v'}:d_K(x_v^s,x_{v'}^s)>0$, \tc{\\}
hence $\phi_s(v)\equiv x_v^s\neq x_{v'}^s\equiv\phi_s(v')$.\\
{\boldmath$(iii)$} $x_v^s\neq x_{v'}^s\forall v,v'\forall s\geq s_{\max}:=\max_{v\neq v'} s_{vv'}<\infty$. \\
{\boldmath$(iv,v)$} From $(ii)$ and the fact that $\phi_s(v)\neq\phi_s(v')$ for some $s$, 
implies $\phi_*(v)\neq\phi_*(v')$ and $\phi(v)\neq\phi(v')$.
\end{proof}

A grand umbrella open problem is whether all metrics with some mild extra conditions 
can be scale-embedded into $\BsdKs$; some results similar to \cref{prop:linfty_embedding}.
For instance if $\cV$ is finite and/or $d_\cV$ is Euclidean.

\begin{open}[Is $d_K$ universal in some interesting sense?]\label{open:dKuniversal}
If all metric spaces that satisfy additional property $\cC$ can be scale-embedded into $\BsdKs$, 
we call $\BsdKs$, $\cC$-universal. Are there natural/permissive/interesting such $\cC$?
\end{open}

See, \cite[Thm.4]{Bennett:98} and \cref{thm:minorize}
for a much easier to satisfy one-sided notion of universality, 
which only requires $d_K$ to be smaller than the source metric $d_\cV$,
i.e.\ $d_K\leqa d_\cV$.

We insist on scale-\emph{equality}, so need to start more modestly
and try to scale-embed specific metric spaces into $\BsdKs$.
Next is an example which will be useful later.
The proof also contains basic ideas in a clean and simple form,
used later again for more complex cases.

\begin{theorem}[$(\SetB^m,d_H)$ scale-embeds into $\BsdKs$]\label{thm:Ham2K}\hfill\par
\noindent There exist string/sequence-embeddings from $\SetB^m$ with Hamming distance $d_H$
into $\SetB^*$ with complexity distance $d_K$.
\end{theorem}

\begin{proof}
Let $v_1,...,v_m,w_1,...,w_m\in\SetB^s$ be 
$2m$ $K$-independent and $K$-random binary strings of length $s$,
i.e.\ $K(v_1...v_m w_1... w_m)\equs 2m\cdot s$, 
and in particular $K(x|y)\equs s$ for any pair of $x\neq y$ from these $2m$ strings.
Such collection of strings exist for $s\gg\log(2m)$.
Now we map $u\in\cU:=\SetB^m$ to $x_{1:sm}=x_1...x_m\in\SetB^{s\cdot m}\subset\SetB^*$ 
via $\phi:\SetB^m\to\SetB^*$ as follows:
Namely, $x_i:=v_i$ if the $i$th bit of $u$ is $0$ and $x_i:=w_i$ if the $i$th bit of $u$ is $1$.
It is easy to see that $d_K(\phi(u),\phi(u')) = s\cdot d_H(u,u') \pm O(\log s)+O(m)$ (Exercise). 
The string-embedding can be converted to a sequence-embedding by considering infinite sequences
$v_i,w_i\in\SetB^\infty$, and instead of concatenating $x_1...x_m$,
we interleave the $m$ sequences, 
or equivalently by defining $\phi(u)_i:=z_{2i-u_{i\%m+1}}$ for $i\in\SetN$ ($\%$=modulo), 
where $z_{1:\infty}$ is one infinite $K$-random sequence, and $f(s,u)=m\cdot s$.
\end{proof}

\Cref{lem:EinHam} implies that trees and $(\SetZ_p)^m$ (for 1-norm $(ii)$ and circular metric $(iv)$) 
are isometrically embeddable into $(\SetB^{pm},\fr12 d_H)$,
hence by \cref{thm:Ham2K} scale-embeddable into $\BsdKs$.
\cref{thm:Ham2K} and its proof generalize 
to hyper-cuboids $(\{0,l_1\}\times...\times\{0,l_m\},||\cdot||_1)$ for $l_i\in\SetR$,
hence, again by \cref{lem:EinHam}, non-negatively weighted trees and $(R_1\times...\times R_m,||\cdot||_1)$ 
for any finite subsets $R_i\subset\SetR$ can also be embedded into $\BsdKs$.
We can actually scale-embed any bounded subset $R$ of $\SetR^m$ with 1-norm into $\BsdKs$.
We start with $\SetR^1$:

\begin{theorem}[$({[a,b],|\cdot|})$ scale-embeds into $\BsdKs$]\label{thm:Interval2K} 
\end{theorem}

\begin{proof} 
$[a,b]$ is isometrically embeddable
iff $[0,b-a]$ is (since $|\cdot|$ is translation invariant)
iff $[0,1]$ is (since scale-embeddings ignore scale).
We cannot simply scale $r\leadsto\fr{r-a}{b-a}$, since $a$ or $b$ may be incomputable reals,
but $r\leadsto(r-\lfloor a\rfloor)/(\lceil b\rceil-\lfloor a\rfloor)$ will do.
So w.l.g.\ consider $r,r'\in[0,1]$ and $r'\geq r$. 
With a little care, we do {\em not} need to assume they are computable.
Let $y_{1:\infty}$ and $z_{1:\infty}$ be two $K$-independent $K$-random sequences.
\\
\emph{String-embedding:} 
Let $t=t_r=\lceil r\cdot s\rceil$ and similarly $t'=t_{r'}=\lceil r'\cdot s\rceil$.
Define the scale-embedding $\phi_s:[0,1]\to\SetB^*$ as
$\phi_s(r):=x^r:=y_{<t}z_{t:s}$ ($\forall r\in[0,1]$). Then
\begin{align*}
  K(x^r|x^{r'}) ~\tc{&}\leqa~ K(z_{t:t'-1},t,t') ~\leqa~ t'-t + 1.1\lb(t) + 1.1\lb(t') ~\tc{\\&}\leqa~ s\cdot(r'-r)+3\lb(s)
\end{align*}
$x^r$ and $x^{r'}$ differ only from bit $t$ to $t'-1$, so an encoding of these bits,
and the locations $t$ and $t'$ allows to reconstruct $x^r$ from $x^{r'}$.
Since $t,t'\in\{0,...,s\}$, they can each be encoded in $1.1\lb s+O(1)$ bits.
Conversely,
$$
    K(x^r|x^{r'}) ~\geqs~ K(z_{t:t'-1}) ~\geqs~ t'-t ~\equa~ s\cdot(r'-r)
$$
since the $K$-random bits $x^r_{t:t'-1}=z_{t:t'-1}$ are independent of all bits in $x^{r'}$.
The logarithmic slack is due to a union bound (see \Cref{sec:RandomEmbedding})
since it needs to hold for all $(s+1)^2$ string pairs $\{(x^r,x^{r'}):r,r'\in[0,1]\}$ simultaneously.
We can show the same bounds for $K(x^{r'}|x^r)$ (with $z_{t:t'-1}$ replaced by $y_{t:t'-1}$), 
hence all 4 bounds together give
$$
    d_K(\phi_s(r),\phi_s(r')) ~\equiv~ d_K(x^r,x^{r'}) ~=~ s\cdot|r'-r| \pm O(\log s)  
$$
hence $\phi_s$ is a scale-embedding.
\\
\emph{Sequence-embedding:} Choose $\phi(r)=x_{1:\infty}^r$,
where $x_i^r=y_i$ if $b_i^r=0$ and $x_i^r=z_i$ if $b_i^r=1$.
We need to choose $b_{1:\infty}^r$ such that the fraction of 1s is $r$,
e.g.\ a Bernoulli($r$) sequence.
But we also need to ensure that $b_i^{r'}\geq b_i^r$ for $r'\geq r$,
so the sequences for difference $r$ cannot be independent.
This can be achieved by sampling $u_i\sim$Uniform$[0,1]$ i.i.d.,
and setting $b_i^r=[\![u_i\leq r]\!]$ ($\forall r\in[0,1]$).
Then $\E[\fr1s\sum_{i=1}^s b_i^r]=r$. We need to derandomize this:
We can choose $u_i=(i\cdot\gamma) ($mod $1)$ for any computable irrational $\gamma$ with bounded partial quotients.
A quantitative version of the equidistribution theorem 
ensures $|\sum_{i=1}^s b_i^r-s\cdot r|\leq 3+c\cdot\lb s$ $\forall r,s$ for some constant $c$,
since $u_i$ is an additive recurrence with low discrepancy $O(\log(i)/i)$ \cite[Chp.2,Thm.3.4]{Kuipers:74}.
E.g.\ $c=2.45$ for $\gamma=(\sqrt{5}-1)/2$. 
Also $b_i^{r'}\geq b_i^r$.
The rest of the proof is now similar to the string-embedding case by 
counting the number of bits in which $x_{1:s}^r$ and $x_{1:s}^{r'}$ differ.
For the upper bound we have
\begin{align*}
    K(x_{1:s}^r|x_{1:s}^{r'}) ~\tc{&}\leqa~ K(x_{1:s}^r \oplus x_{1:s}^{r'}) 
    ~\tc{\\&}\leqa~ s\cdot (r'-r) + 2c\log s + K(r) + K(r')  
\end{align*}
where $\oplus$ is bitwise XOR.
A program which computes $w:=x_{1:s}^r \oplus x_{1:s}^{r'}$ can clearly  
be converted into a program that takes $x_{1:s}^{r'}$ as input, XORs it with $w$,
which results in $x_{1:s}^r$. This proves the first inequality.
As for the second inequality, if we know $r$ and $r'$,
we can compute $b_i^r$ and $b_i^{r'}$ and hence the bits
where $w$ is potentially non-zero ($\Delta:=\{i:b_i^r\neq b_i^{r'}\}$).
$|\Delta| ~=~ \sum_{i=1}^s b_i^{r'}-\sum_{i=1}^s b_i^r ~=~ s\cdot(r'-r)\pm 2c\cdot\log s\pm 6$.
As for the lower bound, string $x_{i\in\Delta}^r$ is $K$-random and $K$-independent of $x_{1:s}^{r'}$,
so its encoding requires at least $|\Delta|$ bits,
leading to 
$$
    K(x_{1:s}^r|x_{1:s}^{r'}) ~\geqs~ |\Delta| ~\equs~ s\cdot(r'-r)
$$
The same arguments lead to the same bounds for $K(x_{1:s}^{r'}|x_{1:s}^r)$.
All 4 bounds together give
\begin{align*}
    d_K(\phi_s(r),\phi_s(r')) ~\tc{&}\equiv~ d_K(x_{1:s}^r,x_{1:s}^{r'}) ~\tc{\\&}\equa~ 
    s\cdot|r'-r| \pm O(\log s) \pm K(r) + K(r')   
\end{align*}
For embedding subsets of $[0,1]$ containing only computable reals, 
we are done, but to allow for incomputable $r$ we have to modify the construction slightly.
We replace $r$ with $r_s:=\lceil r\cdot s\rceil/s$ in the above construction,
and note that $r_s\in\{0,...,s\}/s$ can be encoded in $2.2\lb s+O(1)$ bits.
while it introduces an extra slack of only $s|r_s-r|\leq 1$ bit.
So equality above also holds without the $K(r)+K(r')$ terms as in the string-embedding case.
\end{proof}

One may wonder how an uncountable space $[0,1]$ can be embedded into a countable space $\SetB^*$.
Since a scale-embedding is only approximate for $s<\infty$,
the mapping $\phi_s$ does not need to be injective.
Points closer than $O(1/s)$ can be mapped to the same point.
But any two points $x\neq y$ have distinct images for all sufficiently large $s$.
In a uniform discrete space ($\inf_{x\neq y} d(x,y)>0$)
e.g.\ an $\eps$-grid or $\SetB^\infty$, $\phi_s$ is injective for all sufficiently large $s$.
But if there are cluster=limit=accumulation points 
(they do not need to be in $\cX$ itself, e.g.\ $0\not\in1/\SetN$),
then, while $\phi_s$ separates more and more points for increasing $s$,
for every $s$ there remain points that are merged together.
But at least $\phi_*$ and $\phi$ are always injective (\cref{lem:SEmbedBasic}).

It is also natural to ask about self-embeddings (endomorphisms) and self-isomorphisms (automorphisms).
Are there scale-embeddings from $\BsdKs$ into itself?
Any computable injective function $f:\SetB^*\to\SetB^*$
leaves $K$ approximately unchanged, i.e.\ $K(f(x)|f(y))\equa K(x|y)$ (Exercise),
hence $d_K(f(x),f(y))\equa d_K(x,y)$,
so any injective (bijective) function $f$ is an \emph{approximate} metric endomorphism (automorphism).
Our condition on a scale-embedding is stronger and demands 
$d_K(\phi_s(x),\phi_s(y)) \equs s\cdot d_K(x,y)$ $\forall s\in\SetN$.
We cannot simply concatenate $f(x)$ $s$ times to string $\phi_s(x):=f(x)...f(x)$,
since $K(f(x)...f(x))\equa K(f(x))$ does not scale with $s$.

\begin{open}[Endomorphisms and automorphisms of $\BsdKs$]\label{open:dKauto}
\iftwocol\else\hfil\par\noindent\fi Are there scale-embeddings from $\BsdKs$ into itself?
\end{open}

\section{Random Scale-Embeddings}\label{sec:RandomEmbedding}

In this section we exploit the close relationship 
between data compression and probabilistic modelling:
The optimal expected code-length of a data source $X\sim P$
is its entropy $H(P)=H(X)$. 
Even more, with high probability, 
the optimal code length $K(X)\approx -\lb P(X)$.
We exploit this and define scale-embeddings into string-valued random variables.
We then show that every random scale-embedding can be derandomized,
and also prove a limited converse.
We also prove a product space theorem for random scale-embeddings for $||\cdot||_1$.
As an example application, combining all allows us 
to extend the scale-embedding of $([0,1]^m,||\cdot||_1)$ into $\BsdKs$
from $m=1$ to $m>1$.

\paradot{MDL bound and randomness deficiency} 
For any probability mass function $P:\SetB^*\to[0,1]$,
we can encode $x$ via Shannon-Fano or Huffman or arithmetic encoding in $\lceil-\lb P(x)\rceil$
bits provided we know $P$. This implies 
the MDL bound $K(x)\leqa -\lb P(x)+K(P)$ \cite[Thm.2.7.22]{Hutter:24uaibook2}.
On the other hand, the set of strings that can be compressed significantly better than that 
has small $P$-probability. Formally, let $X=X_{1:n}$ be any string-valued random variable,
then the randomness deficiency $d_P(X):=-\lb P[X]-K(X)\leq \lb(1/\delta)$ 
with probability (w.p.)\ $\geq 1-\delta$ \cite[2.7.29]{Hutter:24uaibook2}.
Together we have that  with high $P$-probability, $K(X)\approx-\lb P[X]$ for computable $P$. 
This motivates to find for some given metric space $(\cV,d_\cV)$,
a $P$ and $X_v$ such that $s\cdot d_\cV(v,v')\approx\max\{-\lb P[X_v|X_{v'}],-\lb P[X_{v'}|X_v]\}$.
Then $\phi_s(v):=x_v$ for any collection of $x_v$ 
satisfying $K(x_v|x_{v'})\approx-\lb P(x_v|x_{v'})$ is a scale-embedding.
When considering $d_K^+$ instead of the default $d_K^\vee$, replace \emph{max} by \emph{average}, here and elsewhere. 

\begin{definition}[Random scale-embeddings into $\BsdKs$]\label{def:RandomInK}
Let $\Phi_s:\cV\times\Omega\to\SetB^*$ be computable random functions, 
so $X_v\equiv X_v^s:=\Phi_s(v,\cdot)$ is a collection of $|\cV|$ (dependent) string-valued random variables.
Let $(\Omega,P)$ be a probability space with computable probability measure, i.e.\ $K(P)=O(1)$.
Assume there exists $c$ such that\ntc{\vspace{-2ex}}
\begin{itemize}\parskip=0ex\parsep=0ex\itemsep=0ex
  \item[$(i)$] $D_{vv'}^s(\omega):=\max\{-\lb P[X_v|X_{v'}],-\lb P[X_{v'}|X_v]\} 
    ~\equs~ s\cdot d(v,v') ~~~\forall v,v',\omega$
  \item[$(ii)$] Let $\dot v_s$ be a representative of $v$ 
  such that $\Phi_s(\dot v_s,\omega)=\Phi_s(v,\omega)$ $\forall\omega\in\Omega$, 
  i.e.\ $X_v\equiv X_{\dot v}$. Let $\dot\cV_s:=\{\dot v_s:v\in\cV\}$ be the set of all such representatives.
  Assume $\exists c,a:|\dot\cV|\leq(s+c)^a$.
\end{itemize}

We call $\Phi_s$ a random scale-embedding if $(i)$ and $(ii)$ are satisfied for all $s$.
For sequence-embeddings we additionally require 
$\Phi_s(v,\omega):=\Phi(v,\omega)_{1:f(s,v)}$, where $\Phi:\cV\times\Omega\to\SetB^\infty$.
In this case $P[X_v|X_{v'}]$ is the probability that infinite sequence $\Phi(v,\omega)$ starts with $X_v$ 
given $\Phi(v',\omega)$ but actually depends only on initial part $X_{v'}$.
If $P[X_v|X_{v'}]=P[X_{v'}|X_v]$ $\forall v,v',s$, we call the embedding symmetric.
\end{definition}

\begin{remark}[Technical remarks on the condition $(i)$ (can be skipped)]\rm
To understand the curious condition $(i)$, 
consider the simple example equation $P[X]=c$, which is short for $P[X=x]=c$ for all $x$ in the range of $X$.
($X$ is a random variable, hence $Y:=f(X)$ is also a random variable, where $f(x):=P[X=x]$, see list of notation,
and $Y=c$ means $Y(\omega)=c$ $\forall\omega$, i.e.\ $Y$ is constant). 
That is, $P$ is a uniform distribution on the range of $X$.
Now consider, $P[X_v]=c_v\forall v,\omega$ for the collection/vector $X_*=(X_v:v\in\cV)$ of all random variables,
and the push-forward measure (also denoted $P$) of $X_*$ on $\Omega_*:=\{X_*(\omega):\omega\in\Omega\}$.
The larger box $\Omega_*=((X_v(\omega):\omega\in\Omega):v\in\cV)$ would work here as well but not for the real case $(i)$.
The conditions mean that the marginals of $P$ are uniform.

The condition $(i)$ can be weakened to hold simultaneously for all $v,v'\in\cV$ (and for $s$ in the sequence-embedding case)
only with some probability $\eps>0$, but this is actually equivalent: Let $\Omega'\subseteq\Omega$
be the set of $\omega$ for which $(i)$ holds. Then it holds for $P'[\cdot]:=P[\cdot|\Omega']$ on \emph{all} $\omega\in\Omega'$,
introducing just another additive constant $\lb P[\Omega']$, that is, $(i)$ holds for $(\Omega',P')$.
We do not even need to exploit the additive slack: 
We can replace $\delta$ by $\eps\cdot\delta$ in the proof of \cref{thm:RandomInK} below,
and apply another union bound.  
\end{remark}

Condition $(ii)$ means that there is the same representative subset $\dot\cV$ of $\cV$ simultaneously for all $\omega$.
Phrased differently: the partition of $\cV$ induced by $\Phi_s(\cdot,\omega)$ is the same for all $\omega$,
and its size needs to be polynomial in $s$.

\begin{lemma}[Derandomization]\label{lem:RandomInK} 
For a random scale-embedding $\Phi_s$, let $\delta_{\dot v\dot v'}>0$ 
such that $\sum_{\dot v,\dot v'\in\dot\cV_s}\delta_{\dot v\dot v'}<1$.
Then there exist $\omega_0^s$ such that for $\phi_s:=\Phi_s(\cdot,\omega_0^s)$ \\
$(i)~$ $|d_K(x_v^s,x_{v'}^s)-D_{vv'}^s(\omega_0^s)| ~~\tc{\\\hspace*{4ex}}\leqa~ \max\{-\lb\delta_{\dot v\dot v'},~-\lb\delta_{\dot v'\dot v},~K(P,\Phi_s)\}$ \\
If furthermore $\Phi_s$ is a sequence-embedding and $\sum_{s\in\SetN}\sum_{\dot v_s,\dot v'_s\in\dot\cV_s}^s\delta_{\dot v_s\dot v'_s}^s<1$, \\
then there exists a single $\omega_0$ such that for $\phi_s:=\Phi_s(\cdot,\omega_0)$ \\
$(ii)$ $|d_\Km(x_v^s,x_{v'}^s)-D_{vv'}^s(\omega_0)| ~\tc{\\\hspace*{4ex}}\leqa~ \max\{-\lb\delta_{\dot v_s\dot v'_s}^s,-\lb\delta_{\dot v'_s\dot v_s}^s,K(P,\Phi)\}$.
\end{lemma}    

For computable $P$ and $\Phi$, which we henceforth assume, 
$K(P,\Phi_s)\leq 2\lb s+O(1)$ and $K(P,\Phi)=O(1)$,
but if $P$ and/or $\Phi$ depend on further parameters, the $O(1)$ becomes $K$(parameters).
For finite $\cV$, $\dot\cV=\cV$ and $\dot v=v$ satisfy the condition on $\dot\cV$.

\begin{proof}
{\bf String-embedding case:} 
For any string-valued random variable $X$ and any computable probability measure $P$ and any $\delta>0$, 
by \cref{lem:Kprop} we have 
\begin{align}\label{eq:randeficientK}
  \lb\delta ~\lesstar~ K(X)+\lb P[X] ~\leqa~ K(P) \ntc{~~~\text{where $~~~\lesstar$ holds w.p.$\geq 1-\delta$}}
\end{align}
\tc{where $\lesstar$ holds w.p.$\geq 1-\delta$.}
We can condition the whole expression on $Y$ (one can regard $Y$ as an oracle) to get
$$
  \lb\delta ~\lesstar~ K(X|Y)+\lb P[X|Y] ~\leqa~ K(P|Y) ~\leqa~ K(P) \ntc{~~~(w.p.\geq 1-\delta)}
$$
Now in our context, we want this to hold for $X_v|X_{v'}$, i.e.\
$$
  \lb\delta_{vv'} ~\lesstar~ K(X_v|X_{v'})+\lb P[X_v|X_{v'}] ~\leqa~ K(P,\Phi_s) \ntc{~~~(w.p.\geq 1-\delta_{vv'})}
$$
\tc{where $\lesstar$ holds w.p.$\geq 1-\delta_{vv'}$.}
The $\Phi_s$ appears, since computing the push-forward measure $P[X_v|X_{v'}]$ requires $P$ \emph{and} $\Phi_s$.
We need this to hold $\forall v,v'\in\cV$ \emph{simultaneously}.
This is already true if it only holds $\forall \dot v,\dot v'\in\dot\cV_s$,
since $\{X_v:v\in\cV\}=\{X_{\dot v}:\dot v\in\dot\cV\}$ by assumption.
Applying a union bound shows that $\lesstar$ holds simultaneously $\forall v,v'\in\cV$
with probability $1-\sum_{\dot v,\dot v'\in\dot\cV_s}\delta_{\dot v\dot v'}>0$.
This means there must be at least one $\omega_0$ (and usually many) such that 
the collection of strings $\{x_v:=\Phi_s(v,\omega_0)\}$ satisfies the above inequalities.
That is,
$$
  \lb\delta_{\dot v\dot v'} ~<~ K(x_v|x_{v'})+\lb P(x_v|x_{v'}) ~\leqa~ K(P,\Phi_s) ~~~\forall v,v'\in\cV
$$
Taking the minimum/maximum (or average in case of $d_K^+$) 
over this and with $v$ and $v'$ swapped in a judicious order
($\delta<K-D$ and $\delta'<K'-D'$ implies $\delta_{\min}<K_{\max}-D_{\max}$) we get
$$
  \min\{\lb\delta_{\dot v\dot v'},\lb\delta_{\dot v'\dot v}\} ~<~ d_K(x_v,x_{v'})-D_{vv'}^s(\omega_0) ~\leqa~ K(P,\Phi_s) \ntc{~~~\forall v,v'\in\cV}
$$
\tc{$\forall v,v'\in\cV$,} which gives $(i)$ after rearranging terms.
\\
{\bf Sequence-embedding case} proceeds similarly but with subtle differences: 
For an infinite random sequence $X_{1:\infty}$, the analogue of \cref{eq:randeficientK} is
\cite[Cor.3.10.3]{Hutter:24uaibook2}
\begin{align}\label{eq:randeficientKm}
  \lb\delta ~\lesstar~ \Km(X_{1:s})+\lb P[X_{1:s}] ~\leqa~ K(P)
\end{align}
where $\lesstar$ holds w.p.$\geq 1-\delta$.
(It actually holds \emph{simultaneously} for all $s$ but curiously we are not able to exploit this).
We can condition this on (oracle) $Y_{1:s}$ to get
$$
  \lb\delta ~\lesstar~ \Km(X_{1:s}|Y_{1:s})+\lb P[X_{1:s}|Y_{1:s}] ~\leqa~ K(P|Y_{1:s}) ~\leqa~ K(P)
$$
Similar to $\Km$, $P$ is actually a measure on infinite sequences. 
Let $X_v^s:=\Phi(v,\cdot)_{1:f(s,v)}$.
Now in our context, we want this to hold for $X_v^s|X_{v'}^s$, i.e.
$$
  \lb\delta_{vv'}^s ~\lesstar~ \Km(X_v^s|X_{v'}^s)+\lb P[X_v^s|X_{v'}^s] ~\leqa~ K(P,\Phi) \ntc{~~~(w.p.\geq 1-\delta_{vv'}^s)}
$$
$w.p.\geq 1-\delta_{vv'}^s$.
The same argument as in $(i)$ shows that this holds 
$\forall v,v'\in\cV$ \emph{and} $\forall s\in\SetN$ simultaneously with non-zero probability,
and hence there must be at least one $\omega_0$ (and usually many) such that 
the collection of strings $\{x_v^s:=\Phi_s(v,\omega_0)\}$ satisfies the above inequalities.
That is,
$$
  \lb\delta_{\dot v_s\dot v'_s}^s ~<~ \Km(x_v^s|x_{v'}^s)+\lb P(x_v^s|x_{v'}^s) ~\leqa~ K(P,\Phi) \ntc{~~~\forall v,v'\in\cV~\forall s\in\SetN}
$$
$\forall v,v'\in\cV~\forall s\in\SetN$.
Taking the minimum/maximum (or average in case of $d_K^+$) over this and with $v_s$ and $v'_s$ swapped, we get
\begin{align*}
  \min\{\lb\delta_{\dot v_s\dot v'_s}^s,\lb\delta_{\dot v'_s\dot v_s}^s\} ~\tc{&}<~ d_\Km(x_v^s,x_{v'}^s)-D_{vv'}^s(\omega_0) 
  ~\tc{\\&}\leqa~ K(P,\Phi) ~~~\forall v,v'\in\cV  
\end{align*}
which gives $(ii)$ after rearranging terms.
\end{proof}

The justification for (the naming of) \cref{def:RandomInK} is the following theorem:

\begin{theorem}[Derandomizing embeddings into $\BsdKs$]\label{thm:RandomInK} 
If there exists a random string/sequence-embedding $\Phi_s$ from $(\cV,d)$ into $\BsdKs$,
then $\phi_s:\cV\to\SetB^*$ with $\phi_s(v):=\Phi_s(v,\omega_0)$ is a conventional string/sequence-embedding
for suitable constant $\omega_0\in\Omega$.
\end{theorem}    

Any $P$-Martin-L\"of-random $\omega_0$ will do \cite[Def.3.10.1]{Hutter:24uaibook2}.
In particular for $\Omega=\SetB^\infty$ and $P(\omega_{1:n})=2^{-n}$,
any $K$-random $\omega_0$ will do. For string-embeddings, $\omega_0$ may (have to) be $s$-dependent.

\begin{proof} 
\emph{String-embedding case:}
$\delta_{\dot v\dot v'}:=\fr12|\dot\cV_s|^{-2}$ satisfies the condition of \cref{lem:RandomInK},
and $-\lb\delta_{\dot v\dot v'}=2\lb|\dot\cV_s|+1=2a\lb(s+c)+1$ by \cref{def:RandomInK}$(ii)$
and $K(P,\Phi_s)\leqa 2\lb s$ for computable $P$ and $\Phi$, hence  
$|d_K(x_v^s,x_{v'}^s)-D_{vv'}^s(\omega_0)|\leqa(2a+2)\lb s$ by \cref{lem:RandomInK}$(i)$.
By \cref{def:RandomInK}$(i)$, $D_{vv'}^s(\omega)\equs s\cdot d(v,v')$ $\forall\omega$.
Together this gives $d_K(x_v^s,x_{v'}^s)\equs s\cdot d(v,v')$,
hence $\phi_s(v):=x_v^s$ satisfies scale-embedding \cref{def:scale_embed}.

The \emph{sequence-embedding case} proceeds similarly but with 
$\delta_{\dot v_s\dot v'_s}^s:=\fr12 s^{-2}|\dot\cV_s|^{-2}$,
leading to an irrelevant extra $2\lb s$ slack, and $K(P,\Phi)=O(1)$,
saving some slack in the other direction.
\end{proof}

\begin{theorem}[Product scale-embeddings into $\BsdKs$]\label{thm:PSEinK} 
If $(\cV,d_\cV)$ and $(\cW,d_\cW)$ have (in case of $d_K^\vee$ symmetric) 
random string/sequence-embeddings $\Phi_s$ and $\Psi_s$ into $\BsdKs$
such that the lengths of $\Phi_s(v,\omega)$ and $\Psi_s(w,\omega)$ are independent of $\omega$,
then so does $(\cV\times\cW,d_\cV+d_\cW)$; and similarly for more than two spaces.
\end{theorem}    

\begin{proof}
For $d_K^\vee$, by assumption, there exists $(\Omega',Q)$ and $\Phi$ and $(\Omega'',R)$ and $\Psi$ such that 
\begin{align*}
  -\lb Q[Y_v|Y_{v'}] ~&\equs~ s\cdot d_\cV(v,v'), ~~~~~\text{where}~~~ Y_v:=\Phi_s(v,\cdot) \\
  -\lb R[Z_w|Z_{w'}] ~&\equs~ s\cdot d_\cW(w,w'), ~~~\text{where}~~~ Z_w:=\Psi_s(w,\cdot)  
\end{align*}
Define the product probability $P:=Q\cdot R$ on $\Omega:=\Omega'\times\Omega''$,
and $\cU:=\cV\times\cW$ hence $\dot\cU=\dot\cV\times\dot\cW$.
The length-condition ensures that joined string-embedding $X_{v,w}:=Y_v Z_w\sim P$ 
is uniquely decodable into strings $Y_v$ and $Z_w$, given $v$ and $w$.
Take the sum of the two equations above. 
Together with the independence of $Q$ and $R$ we get
\begin{align*}
  \ntc{&} -\lb P(X_{v,w}|X_{v',w'}) ~\tc{&}=~ -\lb[Q(Y_v|Y_{v'})\cdot R(Z_w|Z_{w'})] \\
  ~&=~ -\lb Q(Y_v|Y_{v'})-\lb R(Z_w|Z_{w'})
  ~\tc{\\&}\equs~ s\cdot d_\cV(v,v')+s\cdot d_\cW(w,w')
\end{align*}
which, after maxing with the $(v,w)\leftrightarrow (v',w')$-case, 
satisfies the conditions for a symmetric random scale embedding into $(\cU,d_\cU:=d_\cV+d_\cW)$.

The proof for sequence-embeddings is the same, 
except that we have to zip $Y_v$ and $Z_w$ together to preserve sequentiality.
By assumption, $\Phi_s(v)=\Phi(v)_{1:f(s,v)}$
and $\Psi_s(w)=\Psi(w)_{1:g(s,v)}$ for some $f,g$ monotone increasing in $s$.
Then 
$$
  \Xi_s(v,w):=\prod_{t=1}^s \Phi_s(v)_{f(t-1,v)+1:f(t,v)}\Psi_s(w)_{g(t-1,v)+1:g(t,v)}
$$
is a symmetric random sequence embedding $\Xi:\cU\to\SetB^\infty$.

The proof for more than two spaces is similar or by induction.
The proof for $d_K^+$ is very similar (Exercise).
\end{proof}

\begin{open}[Product scale-embeddings into $\BsdKs$]\label{open:PSEinK}
Can the conditions in \cref{def:RandomInK} be weakened and still preserve \cref{thm:RandomInK}?
Or does \cref{thm:PSEinK} even hold for all conventional embeddings with no further assumptions at all?
\end{open}    

As for the last question, without the auxiliary probability distribution $P$,
the question is how to make the embeddings $\phi$ and $\psi$ from two metric spaces $\cV$ and $\cW$ $K$-independent.
We need to jointly scramble all $y\in\cY:=\phi(\cV)\subset\SetB^*$ to $\tilde y$ 
(and possibly the same for $x\in\cX:=\psi(\cW)\subset\SetB^*$) in such a way as to 
$(i)$ preserve $K(y|y')\equa K(\tilde y|\tilde y')$, and 
$(ii)$ make $\tilde y\tilde y'$ $K$-independent of $xx'$ $\forall x,x'\in\cX$ and $\forall y,y'\in\cY$,
to ensure $K(x\tilde y|x'\tilde y')\equa K(x|x')+ K(\tilde y|\tilde y')$.
We actually only need the weaker condition 
$d_K(x\tilde y,x'\tilde y')\equs d_K(x,x')+ d_K(\tilde y,\tilde y')$
to prove a general non-random version of \cref{thm:PSEinK}. 
The assumption about the existence of $P$ with the property given in 
\cref{def:RandomInK}$(i)$ enables such scrambling.

Instead of trying to derive a non-random version of \cref{thm:PSEinK},
we could seek for a reverse of \cref{thm:RandomInK},
i.e.\ randomizing a non-random scale-embedding.
The following theorem and discussion thereafter could be skipped,
since it is not used elsewhere.
On the other hand, it is quite interesting in its own,
an improvement on Romashchenko's typization trick \cite[Thm.211]{Shen:17}.

\begin{theorem}[General randomization]\label{thm:randomization}
Let $x_1,...,x_m\in\SetB^*$ be a collection of strings,
and $x_I:=(x_i)_{i\in I}$ for $I\subseteq M:=\{1,...,m\}$ be a sub-collection, and similarly $x_J$.
Let $\cI\subseteq 2^M$ be a collection of such $I$ that includes $M$.
Then there exist random variables $Z_1,...,Z_m$ 
such that $|\lb P[Z_I|Z_J]+K(x_I|x_J)|\leqa (2.2|\cI|+5.5)\lb s+2.2K(\cI)$ for all $\omega\in\Omega$,
where $s:=K(x_M)$, simultaneously for all $I,J$ for which $I\cup J,J\in\cI$. 
If $J$ is empty, then the condition $Z_\emptyset=\emptyset$ 
and $x_\emptyset=\emptyset$ can be dropped.
In particular $P[Z_I|Z_J]$ is approximately uniform.
\end{theorem}

The basic idea of the proof is to construct a uniform distribution
$P[Z_M=z_M]=1/|A|$ $\forall z_M\in A$ and $0$ otherwise,
where $A\subset (\SetB^*)^m$ is the set of all $z_M$ such that 
$K(z_I|z_J)\approx K(x_I|x_J)$ for the desired $I,J$,
and then show that $P(z_I|z_J)\approx 2^{-K(x_I|x_J)}$.

In \cite{Hammer:00}, a superficially related but actually quite different question is asked:
Do all Shannon-inequalities also hold for Kolmogorov complexity and vice vera.
The flavor is: If some statement is true for \emph{all} strings, 
is it then also true for all random variables.
As an example $K(x)+K(y)\gtrsim K(x,y)$ 
holds for all strings and $H(X)+H(Y)\geq H(X,Y)$ holds for all random variables.
On the other hand, \emph{we} are asking whether for a particular set of strings $x_1,...,x_m$ with 
particular non-universal (metric $d_K$) properties, 
there are corresponding random variables with the same (metric $D_{vv'}^s$ or $d_I$) property,
the reverse of \cref{thm:RandomInK}.
Nevertheless many ideas carry over.

\Cref{thm:randomization} strengthens/extends \cite[Thm.211]{Shen:17} in various ways:
First, they use $\cI=2^M$, hence their bound is (always) exponential in $m$,
while we only need $|\cI|=O(m^2)$, an exponential improvement.
More importantly, they only impose \emph{bounds} $K(z_I|z_J)\leq K(x_I|x_J)$ on all \emph{conditionals},
which suffices to prove their results for set sizes (projections and \emph{maximal} slices of $A$),
but not for showing that all slices have the same size, which we need for $P$. 

\begin{proof}
Let $x_M$ be a fixed collections of strings.
Note that $2^{-O(1)}=2^{-K(\epstr)}<\sum_{z\in\SetB^*} 2^{-K(z)}<1$ by Kraft inequality.
Let $\stackrel{c}=$ denote ``equal within $\pm c\lb s\pm O(1)$'' for complexities
and ``equal within $\times 2^{\pm c\lb s\pm O(1)}$'' for probabilities,
where $O(1)$ may depend on further parameters like $m$.
With this notation, $K(x,y)\stackrel{1.1}= K(x|y)+K(y)$.

{\boldmath\bf $A_b$:}
Let $\cI\subset 2^M$ be some collection of $I$ that
includes $I\cup J$ and $J$ and $M$.
Define 
\begin{align*}
  A_b ~&:=~ \{z_M:K(z_I)\leq K(x_I)-b~\forall I\in\cI\} ~~~\text{for}~~~ b\geq 0\\  
  |A_b| ~&\leq~ |\{z_M: K(z_M)\leq K(x_M)-b\}| ~\leq~ 2^{K(x_M)-b}
\end{align*}
by Kraft inequality. On the other hand, $x_M\in A_0$,
so we can encode $x_M$ in $K(x_M)\leq\lb|A_0|+K(A_0)+c_U$ bits.
Ultimately we are interested in $A \equiv A_M ~:=~ A_0\setminus A_b$,
which consists of all and only collections $z_M$ for which $K(z_I)\approx K(x_I)$.

{\boldmath\bf Estimating the size of $A$:}
We first need to estimate the size of $A$.
We cannot use it directly, since $A$ is not computably enumerable, but $A_0$ is.
To describe $A_0$ we need to know $K(x_I)$ $\forall I\in\cI$ and $\cI$,
hence 
$$
  c_s ~:=~ K(A_0) ~\leqa~ \sum_{I\in\cI} K(K(x_I))+K(\cI) ~\leqa~ 1.1|\cI|\lb s+K(\cI)
$$
Typically $\cI$ has a very simple structure, 
e.g.\ $\cI=2^M$ as in \cite{Hammer:00} or just $\cI=\{I\cup J,J,M\}$ or 
in this work, $\cI=\{I:|I|\leq 2\}\cup\{M\}$ with $K(\cI)\leqa 2\log m$.
Worst case is $K(\cI)\leqa 2m|\cI|$.
Together we get
\begin{align} \nonumber
  2^{K(x_M)} ~&\geq~ |A_0| ~\geq~ |A| ~=~ |A_0\setminus A_b| ~=~ |A_0|-|A_b| \\
  ~&\geq~ 2^{K(x_M)-c_s-c_U}-2^{K(x_M)-b} ~=~ 2^{K(x_M)-b} \label{eq:Alowbnd}
\end{align}
for $b:=c_s+c_U+1\leqa 1.1|\cI|\log s$. 

{\boldmath\bf Marginals $A_J$ and slices $A_{M|z_J}$:}
In order to marginalize and condition $P$,
we need to consider projections and slices of $A$. Let 
$$
  A_{I|z_J} ~:=~ \{z_I: \exists z_{\overline{I\cup J}}: z_M\in A\}
  ~~~\text{where}~~~ \bar J:=M\setminus J
$$
be a slice through $z_J$ orthogonal to the $J$-coordinates and projected onto the $I$-coordinates.
If it helps intuition, identify $\SetB^*$ with $\SetN$ and $\epstr$ with $0$,
then $z_M$ is a vector in $\SetN^m$ with coordinate $i\in M$ now being $z_i\in\SetN$.
$A_{M|z_J}$ is an axis-parallel slice of all $z_M$ but with $z_J$ fixed,
while projection $A_I\equiv A_{I|\emptyset}$ keeps only the $I$-coordinates of each $z_M\in A$.
Next we need to bound $|A_{I|z_J}|$ as we did with $|A|$.
Define $A_{I|z_J}^b ~:=~ \{z_I: \exists z_{\overline{I\cup J}}: z_M\in A_b\}$ similarly.
The primary difference here is that $A_{I|z_J}^0\setminus A_{I|z_J}^b\neq A_{I|z_J}$ is possible,
since projection and set intersection do not commute.
We still have $A_{I|z_J}^b\subseteq A_{I|z_J}^0$ and 
$A_{I|z_J}^0\setminus A_{I|z_J}^b\subseteq A_{I|z_J}$.
Then following the same steps as \cref{eq:Alowbnd} gives the lower bound 
$|A_{I|z_J}|\geq 2^{K(x_I|x_J)-b}$, 
potentially with an increased universal constant $c_U$ in $b$.
The upper bound requires a bit of extra work:
For $z_M\in A$ we have $K(z_{I\cup J})\leq K(x_{I\cup J})$ 
and $K(x_J)\not\leq K(z_J)-b$ (since $z_M\not\in A_b$),
hence 
\begin{align*}
  K(z_I|z_J) ~&\stackrel{1.1}=~ K(z_{I\cup J})-K(z_J) 
  ~\tc{\\&}\leq~ K(x_{I\cup J})-K(x_J)+b ~\stackrel{1.1}=~ K(x_I|x_J)+b \\
  \Longrightarrow~~~ |A_{I|z_J}| ~&\leq~ \{z_I:K(z_I|z_J) \stackrel{2.2}\leq K(x_I|x_J)+b\}
  ~\tc{\\&}\stackrel{~1.1|\cI|+2.2~}\leq~ 2^{K(x_I|x_J)}
\end{align*}

{\boldmath\bf $P(z_M)=1/|A|:$}
A natural idea now is to define a uniform distribution on $A$:
$$
  P(z_M) ~:=~ 
  \begin{cases} 1/|A| & \text{if}~~ z_M\in A \\ 0 & \text{else} \end{cases}
$$
The conditionals of a uniform distribution are always also uniform
($P(z_{\bar J}|z_J)=1/|A_{\bar J|z_J}|$),
but the marginals are in general not, but for $A$ they are approximately.

{\boldmath\bf $P(z_J)$ and $P(z_I|z_J)$:} For $z_M\in A$ we now can derive
\begin{align*}
  P(z_J) ~&=~ \sum_{\nq\nq z_{\bar J}\in(\SetB^*)^{|\bar J|}\nq\nq} P(z_M)
  ~=~ \sum_{z_{\bar J}\in A_{\bar J|z_J}\nq}\frac{1}{|A|}
  ~=~ \frac{|A_{\bar J|z_J}|}{|A|} ~\tc{\\&}\stackrel{~2.2|\cI|+2.2~}=~ \frac{2^{K(x_{\bar J}|x_J)}}{2^{K(x_M)}}
  ~\stackrel{1.1}=~ 2^{-K(x_J)} \\
  P(z_I|z_J) ~&=~ \frac{P(z_{I\cup J})}{P(z_J)} 
  ~\stackrel{\ntc{~}2.2|\cI|+4.4\ntc{~}}=~ \frac{2^{-K(x_{I\cup J})}}{2^{-K(x_J)}} ~\stackrel{1.1}=~ 2^{-K(x_I|x_J)}
\end{align*}
Taking the logarithm, concludes the proof.

{\boldmath\bf Remark:}
In our notation, \cite[Thm.211]{Shen:17} reads
$\lb\,\max_{z_J}|A^0_{I|z_J}|\approx K(x_I|x_J)$,
while we have achieved the stronger
$\lb|A_{I|z_J}|\approx K(x_I|x_J)~\forall z_J$,
which enabled proving \cref{thm:randomization}.
\end{proof}

We can use \cref{thm:randomization} to randomize scale-embeddings:

\begin{corollary}[Randomizing scale-embeddings]\label{cor:SE2randSE}
Let $\phi_s$ be a scale-embedding from $\cV$ of size $m$ into $\BsdKs$.
Then there exists a random string-embedding from $\cV$ to $\BsdKs$
with $|D_{vv'}^s - d_K(x_v,x_{v'})|=O(m^2\log s)$.
\end{corollary}

\begin{proof}
Assume $(\cV,d_\cV)$ with $|\cV|=m$ scale-embeds into $\BsdKs$,
that is $s\cdot d_\cV(v,v')\equs d_K(x_v,x_{v'})$ for some $x_v=\phi_s(v_s)$.
By \cref{thm:randomization} there exist random variables $X_v$
such that $K(x_v|x_{v'})\equs-\lb P[X_v|X_{v'}]$. 
Symmetrizing we get $s\cdot d_\cV(v,v')\equs d_K(x_v,x_{v'})\equs D_{vv'}^s$,
hence $\Phi_s(v,\omega):=X_v$ is a random scale-embedding.
Choosing $\cI=\{I\subseteq M:1\leq|I|\leq 2\}$ proves 
the (hidden) precision $m^2$ in the $O()$ term.
\end{proof}

Chaining \cref{cor:SE2randSE,thm:PSEinK,thm:RandomInK} 
we can get a non-random version of the product embedding \cref{thm:PSEinK},
albeit only for string-embeddings of finite spaces with bad constants.

\begin{proof}
W.l.g.\ let $\cV=\{1,...,m\}$. In \cref{thm:randomization} choose $x_i=\phi_s(i)$
and $\cI=\{I\in 2^M:1\leq |I|\leq 2\}\cup\{M\}$. 
Then $|\cI|=\fr12 m(m+3)$. Let $X_i=:\Phi_s(i,\cdot)$ ($Z_i$ in \cref{thm:randomization}) be 
the collection of random variables 
such that $P[X_M=z_M]=P(z_M)$ with $P$ from the theorem/proof.
Choose $I=\{i\}$ and $J=\{j\}$. Clearly $I\cup J\in\cI$.
With this notation, \cref{thm:randomization} implies 
$$
  |\lb P[X_i|X_j]+K(x_i|x_j)| \ntc{~}\leqa\ntc{~} (1.1m(m+3)+5.5)\lb 2+K(\cI)
$$
(formally $\Omega=A$ and $Z_M(\omega)=\omega$)
where $K(\cI)\leqa 1.1\lb m$ is negligible.
We get a similar inequality for $i\leftrightarrow j$.
Taking the max (or average in case of $d_K^+$) of both we get
$|D_{ij}-d_K(x_i,x_j)|=O(m^2\lb s)$.
\end{proof}

For finite spaces, $O(m^2\log s)$ may be an acceptable constant,
but it is exponentially larger than in the derandomization \cref{lem:RandomInK}
(for $\delta_{\dot v\dot v'}=1/m^2$ the slack becomes $2\max\{\lb m,\lb s\}$).
If we tried to randomize $\phi_s:[0,1]\to\BsdKs$ 
of \cref{thm:Interval2K} via \cref{cor:SE2randSE},
$m=|\dot\cV|=s+1$, and the bound becomes vacuous.
Luckily a tailored randomized version work very well:

\begin{corollary}[1-norm scale-embeddings into $\BsdKs$]\label{cor:OneNormEinK}
Any bounded subset of $\SetR^m$ with $||\cdot||_1$ 
can be string/sequence-embedded into $\BsdKs$.
\end{corollary}    

\begin{proof}
Let $\cV$ be a bounded subset of $\SetR^m$. Then $\cV\subseteq[a,b]^m$ for some $a<b$.
By \cref{thm:Interval2K}, $[a,b]$ can be scale embedded into $\BsdKs$.
W.l.g.\ assume $\cU=[a,b]=[0,1]$.\\
\emph{String-embedding case:}
Consider the scale-embedding $\phi_s$ for $[0,1]$ constructed in \cref{thm:Interval2K},
but replace the fixed $K$-random strings $y_{1:s}$ and $z_{1:s}$ by 
$2s$ independent Bernoulli($\fr12$) random variables $Y_{1:s}$ and $Z_{1:s}$.
That is $\Omega:=\{0,1\}^{2s}$, $P[\omega]:=2^{-2s}$ and $Y_i(\omega)=\omega_i$ and $Z_i(\omega)=\omega_{s+i}$,
and $\Phi_s(r,\cdot):=X^r:=Y_{<t}Z_{t:s}$ with $t=\lceil r\cdot s\rceil$ as before.
Then $P[X^r|X^{r'}]=2^{-|t'-t|}\lessgtr 2^{-s|r'-r|\pm 1}$, 
since $X^{r'}$ fixes all bits of $X^r$ except $|t'-t|$ bits.
This satisfies \cref{def:RandomInK}$(i)$.
The embedding $\Phi_s(\cdot,\omega)$ maps all $r\in(\fr{i}{s},\fr{i+1}{s}]$ to the same string,
and this is true for all $\omega\in\Omega$.
Hence we can choose $\dot r_s=\lceil r\cdot s\rceil/s$,
so $\dot\cU=\{0,\fr1s,\fr2s,...,\fr{s-1}{s},1\}$, which satisfies \cref{def:RandomInK}$(ii)$.
Considering $m$ replications of this construction for $u_i\in\cU_i\subseteq[0,1]$,
by \cref{thm:PSEinK}, $(\cV=[0,1]^m,||\cdot||_1)$ random 
scale-embeds into $\BsdKs$,
since $|u_1-u'_1|+...+|u_m-u'_m|=||u_{1:m}-u'_{1:m}||_1$.
Now by \cref{thm:RandomInK} with $|\dot\cV|=|\dot\cU^m|=(s+1)^m$ in \cref{def:RandomInK}$(ii)$
it also conventionally scale-embeds into $\BsdKs$.
Hence $(\cV,||\cdot||_1)$ scale-embeds into $\BsdKs$.
\\
\emph{Sequence-embedding case:} By \cref{thm:Interval2K}, 
$\cU=[0,1]$ can be sequence-embedded into $\BsdKs$.
We can randomize this $\phi$ in the same way by replacing 
$K$-random $y_{1:\infty}$ and $z_{1:\infty}$,
by (infinitely many) independent Bernoulli($\fr12$) random variables $Y_{1:\infty}$ and $Z_{1:\infty}$,
Then $P[X^r_{1:s}|X^{r'}_{1:s}]=2^{-|\Delta|}=2^{-s|r'-r|\pm 5\log s\pm 6}$. 
This satisfies \cref{def:RandomInK}$(i)$,
and $\dot\cU$ is the same as above, so \cref{def:RandomInK}$(ii)$ is satisfied too.
Hence by applying \cref{thm:PSEinK,thm:RandomInK}, 
$(\cU,||\cdot||_1)$ and hence $(\cV,||\cdot||_1)$ sequence-embed into $\BsdKs$.
\end{proof}

As a start, we can construct embeddings into $(\SetR^m,||\cdot||_1)$ and then apply \cref{cor:OneNormEinK},
but unfortunately this approach only works for spaces with less than 5 points:

\begin{proposition}[Only 4 points iso-embed into $(\SetR^3,||\cdot||_1)$]\label{prop:Embed4in1Norm} 
$(i)$ Every metric space of size $m\leq 4$ can be isometrically embedded into $(\SetR^{m-1},||\cdot||_1)$.
($ii)$ There exist metric spaces of size 5 that cannot be iso-embedded into $\ell_1$.
\end{proposition}

\begin{proof}
{\boldmath$(i)$} Consider $(\{v_0,v_1,v_2,v_3\},d)$. Abbreviate $d_{ij}=d(v_i,v_j)$.
W.l.g.\ permute the points so that $d_{02}+d_{13}$ is largest among all possible permutations of indices $0123$,
i.e.\ such that $d_{02}+d_{13}\geq d_{01}+d_{23}$ and $d_{02}+d_{13}\geq d_{03}+d_{12}$.
We now embed one point after the other trying to match the distance requirements,
mapping $\phi(v_0)=(0,0,0)$, $\phi(v_1)=(x_1,0,0)$, $\phi(v_2)=(x_2,y_2,0)$, and $\phi(v_3)=(x_3,y_3,z_3)$,
where we assume $x_1,x_2,y_2,x_3,z_3\geq 0\geq y_3$ and $x_1\geq x_2\geq x_3$.
Under these assumptions, the absolute bars in the distances 
$d_{ij}=|x_i-x_j|+|y_i-y_j|+|z_i-z_j|$ can be eliminated and we get the solvable linear equation system
\begin{align*}
  \tc{&}\scriptsize\arraycolsep=-1pt \left(\begin{array}{c} d_{01} \\ d_{02} \\ d_{12} \\ d_{03} \\ d_{13} \\ d_{23} \end{array}\right) =
  \arraycolsep=2pt\left(\begin{array}{ccc|ccc}
    1 & 0 & 0 & 0 & 0 & 0 \\
    0 & 1 & 1 & 0 & 0 & 0 \\
    1 &-1 & 1 & 0 & 0 & 0 \\ \hline
    0 & 0 & 0 & 1 &-1 & 1 \\
    1 & 0 & 0 &-1 &-1 & 1 \\
    0 & 1 & 1 & -1& 1 & 1 \\
  \end{array}\right)\!\!
  \left(\begin{array}{c} x_1 \\ x_2 \\ y_2 \\ x_3 \\ y_3 \\ z_3 \end{array}\right)
\tc{\\[1ex]}~~~\Longrightarrow~~~
\tc{&}\scriptsize\arraycolsep=-1pt \left(\begin{array}{c} x_1 \\ x_2 \\ y_2 \\ x_3 \\ y_3 \\ z_3 \end{array}\right) = \frac12\!
  \arraycolsep=2pt\left(\begin{array}{ccc|ccc}
    2 & 0 & 0 & 0 & 0 & 0 \\
    1 & 1 &-1 & 0 & 0 & 0 \\
   -1 & 1 & 1 & 0 & 0 & 0 \\ \hline
    1 & 0 & 0 & 1 &-1 & 0 \\
    1 &-1 & 0 & 0 &-1 & 1 \\
    0 &-1 & 0 & 1 & 0 & 1 \\
  \end{array}\right)\!\!
  \arraycolsep=-1pt \left(\begin{array}{c} d_{01} \\ d_{02} \\ d_{12} \\ d_{03} \\ d_{13} \\ d_{23} \end{array}\right)  
\end{align*}
That the solution vector indeed satisfies the 8 assumptions
follows (by explicit verification) from the triangle inequalities that $d$ satisfies 
and the relations due to judicious sorting,
so this is indeed a valid isometric embedding.
\\
{\boldmath$(ii)$} See end of \cref{sec:Euclid}.
\end{proof}

\begin{proposition}[All 5-point metrics scale-embed into $\BsdKs$]\label{prop:point5indK}\label{prop:K32indK}
All 5-point metrics sequence-embed into $\BsdKs$.
Furthermore all connected 5-node graphs with path distance except $K_{3,2}$ 
iso-embed into $(\SetR^3,||\cdot||_2)$ or $(\SetR^{10},||\cdot||_1)$.
\end{proposition}

\begin{proof}
{\bf\boldmath $|G|=5$:} There are 21 undirected unlabelled connected graphs with 5 nodes.
It is straightforward to see by inspection that they either embed into $(\SetR^3,||\cdot||_2)$
(and hence $\ell_1$ by \cite[Ex.15.5.2b]{Matousek:02})
or into $(\SetB^4,d_H)$ or $(\SetB^4,\fr12 d_H)$ (cf.\ \cref{lem:EinHam}) and hence 
$(\SetR^d,||\cdot||_1)$, and $d\leq({5\atop 2})$ by \cite[Ex.15.5.2c]{Matousek:02}.
Their sequence-embedding into $\BsdKs$ now follows from \cref{cor:OneNormEinK}.

{\bf\boldmath $K_{3,3}$:}
Let $T,U,V,W$ be independent Bernoulli($\fr12$) random variables.
Define $\Phi_1(i)=X_i$, where $i\in\cV=\{1,2,3,4,5,6\}$ are the vertices in $K_{3,3}$
and $X_1=TU$, $X_2=VW$, $X_3=(T\oplus V)(U\oplus W)$, $X_4=TV$, $X_5=UW$,
and $X_6=(T\oplus U)(V\oplus W)$.
Then it is easy to see that $D_{ij}^1=-\lb P[X_i|X_j]=-\lb P[X_j|X_i]=(d_{K_{3,3}})_{ij}$ (Exercise).
Now if we use $s$ independent copies of this construction,
we get $\Phi_s(i)=X_i^1...X_i^s$ and $D_{ij}^s=s\cdot D_{ij}^1=s\cdot (d_{K_{3,3}})_{ij}$.
Finally we derandomize with \cref{thm:RandomInK}.

{\bf\boldmath General 5-point metrics:} 
An $m$-point metric $d_{ij}:=d_M(i,j)$ over $M=\{1,...,m\}$ is determined by 
$n=({m\atop 2})$ numbers $(d_{ij})_{1\leq i<i\leq m}$.
W.l.g.\ we may assume $\sum_{i<j} d_{ij}=1$.
Let $\Lambda_m:=\{d_M:d_{ij}\geq 0\wedge d_{ik}\leq d_{ij}+d_{kj}\forall i,j\}$ 
be the set of all semi-metrics (\cref{def:distance}).
If we regard $d$ as a vector in $\SetR^n$, 
then $\Lambda_m\subset\SetR^n$ is a $n-1$-dimensional convex polytope.
Every point in a polytope can be described by a convex combination of its corners $\cC$.
If we can represent each corner metric $d^c$ (called extreme ray metrics)
via a random vector $X^c$ with $H(X_i^c|X_j^c)=d_{ij}^c$,
and $X:=(X^c)_{c\in\cC}$ is a collection of independent random vectors, 
then $d_{ij}=\sum_{c\in\cC} \lambda_c d_{ij}^c = \sum_{c\in\cC} H(X_i^c|X_j^c) = H(X_i|X_j)$ for some $\lambda_{ij}\geq 0$.
So we only need to check iso-embeddability of the extreme ray semi-metrics.
For small $m$, we can find all corners of numerically.

For $m=5$, the $14$-dimensional polytope $\Lambda_m$ has $25$ extreme ray semi-metrics,
but only $3$ non-isometric ones \cite{Grishukhin:92}.
Among them are our good friend $K_{3,2}$ and two so-called 
cut metrics $\text{Cut}(1,4)$ and $\text{Cut}(3,2)$, 
which trivially embed via $X_1=U$ and $X_2=X_3=X_4=X_5=V$ in the former, 
and $X_1=X_2=U$ and $X_3=X_4=X_5=V$ in the latter case,
where $U,V$ are independent Bernoulli($\fr12)$.
\end{proof}

\section{\boldmath Negative Embeddability Results into $\BsdKs$}\label{sec:KraftIneq} 

In this section, after developing some general negative embeddability result 
based on the bounded neighborhood size of $d_K$, 
we will show that not all finite metrics scale-embed into $\BsdKs$.

In order to prove general negative results about embeddings of some $(\cV,d_\cV)$ into $\BsdKs$,
we need to look at special properties which $d_K$ has but $d_\cV$ lacks.
One key property implied by Kraft inequality is:
\begin{equation}\label{eq:Kraft}
    \sum_{x\in\SetB^*} 2^{-d_K^\vee(x,y)} ~\leq~1  
\end{equation}
It is also true for $2d_K^+$. So here the $\fr12$ in the definition is mildly annoying (which we ignore for now).
This implies that any $d_K$-ball of radius $r$ contains at most $2^r$ points:
$$
    B_r^K(x) ~:=~ \{y\in\SetB^*:d_K(x,y)\leq r\} ~~~\stackrel{\cref{eq:Kraft}}\Longrightarrow~~~ |B_r^K(x)| ~\leq~ 2^r
$$
Many but not all practical metrics for discrete spaces have finite neighborhood 
which also does not grow super-exponentially, and virtually all are (upper-semi)computable,
and $d_K$ minorizes these so-called \emph{admissible} distances \cite[Sec.IV]{Bennett:98}:

\begin{theorem}[$d_K$ minorizes all admissible distances $d_\cV$]\label{thm:minorize}
For any symmetric (upper-semi)computable distance $d_\cV$ with countable $\cV$
and $\sum_{v\in\cV} 2^{-d_\cV(v,v')} < \infty$,
and computable encoding $\phi:\cV\to\SetB^*$,
we have $d_K(\phi(v),\phi(v'))\leqa d_\cV(v,v')$ $\forall v,v'\in\cV$.
Conversely, $d_K\leqa d_\cV$ implies $\sum_{v\in\cV} 2^{-d_\cV(v,v')} < \infty$.
\end{theorem}
\begin{proof}
Let $P(x|x'):=c\cdot 2^{-d_\cV(v,v')}$ for the, say, lexicographically first $v,v'\in\cV$
for which $x=\phi(v)$ and $x'=\phi(v')$,
and $P(x|x'):=0$ if no such $v,v'$ exist.
This implies $\sum_{x\in\SetB^*}P(x|x')\leq 1$ for a suitable choice of (computable) $c$,
which makes $P$ lower-semicomputable.
An $x'$-conditioned version of \Cref{lem:Kprop}$(iii)$ now gives
$K(x|x')\leqa -\lb P(x|x')+K(P|x')\equa d_\cV(v,v')-\lb c$.
Taking the maximum of this and $(x,v)\leftrightarrow (x',v')$,
we get $d_K(\phi(v),\phi(v'))\leqa d_\cV(v,v')$ by symmetry of $d_\cV$.
For the converse, $\sum_v 2^{-d_\cV(v,v')}\leqa \sum_x 2^{-d_K(x,x')} \leq \sum_x 2^{-K(x|x')} \leq 1$.
\end{proof}

\Cref{thm:minorize,eq:Kraft} place some limitations 
on what metric spaces $(\cV,d_\cV)$ can be embedded into $\BsdKs$.
Consider the $d_\cV$-balls $B_r(v):=\{v'\in\cV:d(v,v')\leq r\}$. 
If $\phi_s$ is a scale-embedding, then for $s>1$ and large-enough $c$,
\begin{align}\nonumber
  & \phi_s(B_r(v)) ~\subseteq~ \{\phi_s(v'):d_K(\phi_s(v),\phi_s(v'))\leq s\cdot r+c\log s\} \\
  ~&\subseteq~ B_{sr+c\log s}^K(\phi_s(v)) 
  \text{~hence}~ |\phi_s(B_r(v))| ~\leq~ 2^{sr+c\log s} \label{eq:EmbedKraft}
\end{align}
but $\phi_s$ may not be injective (see discussion after \cref{thm:Interval2K}), 
so $|B_r(v)|>|\phi_s(B_r(v))|$ is still possible, so this bound is no constraint.
We can get a constraint by considering $\eps$-balls:

\begin{lemma}[Large neighborhoods prevent embeddings into $\BsdKs$]\label{lem:KraftIneq}
For some $\eps>0$, let $\dot\cV\subset\cV$ be the centers of non-overlapping $\eps$-balls in $(\cV,d_\cV)$.
If there exists (for some center $\dot v\in\dot\cV$ and some radius $r$) a $\dot\cV$-ball $\dot B_r(\dot v)\subseteq\dot\cV$ 
with more than $2^{sr+c\log s}$ points for $s>c\log s/2\eps$, 
then $(\cV,d_\cV)$ cannot be scale-embedded into $\BsdKs$.
\end{lemma}

\begin{proof}
Let $\phi_s$ be a scale-embedding from $(\cV,d_\cV)$ into $\BsdKs$.
Let $\dot\cV\subset\cV$ be the centers of non-overlapping $\eps$-balls in $\cV$, 
i.e.\ $B_\eps(\dot v)\cap B_\eps(\dot v')=\{\}$ $\forall \dot v\neq \dot v'\in\dot\cV$.
Let $\dot B_r(\dot v):=B_r(\dot v)\cap\dot\cV $ be balls in $\dot\cV$.
Then 
$$
    d_K(\phi_s(\dot v),\phi_s(\dot v')) ~\geq~ s\cdot d(\dot v,\dot v') -c\log s ~\geq~ s\cdot 2\eps-c\log s ~>~ 0 \ntc{~~~\text{if}~~~ s>(c\log s)/2\eps}
$$
\tc{if $s>(c\log s)/2\eps$.}
This implies that $\phi_{s|\dot\cV}$ is injective, so by \eqref{eq:EmbedKraft} every embedding into $d_K$ must satisfy
$$
    |\dot B_r(\dot v)| ~=~ |\phi_s(\dot B_r(\dot v))| ~\leq~ 2^{sr+c\log s} ~~~\text{for}~~~ s>(c\log s)/2\eps 
$$\iftwocol\\[-3ex]\else\\[-7ex]\fi
\end{proof}

Consider for example grid $\dot\cV=(2\eps\SetZ)^m\subset\SetR^m=\cV$ with $p$-norm. 
Then open $\eps$-balls centered at $\dot\cV$ are disjoint
(or use closed $\eps'$-balls with $\eps'$ slightly smaller than $\eps$),
and note that 
$$
    |\dot B_r^p(\dot v)| ~\leq~ |\dot B_r^\infty(\dot v)| ~\leq~ (\fr{r}{\eps}+1)^m ~\leq~ 2^{sr+c\log s}
$$
for sufficiently large $c$. This means Kraft's inequality in itself is no obstacle 
in potentially embedding $(\SetR^m,||\cdot||_p)$ nor $(\SetZ^m,||\cdot||_1)$.

\begin{open}[Can $(\SetZ^m,||\cdot||_1)$ be scale-embedded into $\BsdKs$?]\label{open:dKgrid}
\end{open}

From \cref{lem:EinHam}$(vii)$ and \cref{thm:Ham2K} 
we know that we can scale-embed any \emph{bounded} subset of $\SetZ^m$ into $\BsdKs$.
On the other hand, consider some infinite-dimensional (vector) spaces:

\begin{proposition}[Many $\infty$-dim.\ spaces do not scale-embed into $\BsdKs$]\label{prop:InfDimNotDK}
None of $\{0^i10^\infty:i\in\SetN_0\}\subset\SetB^\infty\subset{[0,1]}^\infty\subset\SetR^\infty$ 
with $1\leq p\leq\infty$-norm scale-embed into $\BsdKs$.
\end{proposition}

\begin{proof}
It suffices to prove the statement for the subspace of one-hot vectors (or sequences)
$\cV=\{0^i10^\infty:i\in\SetN_0\}$.
It has disjoint open $\fr12$-balls due to $d_\cV(v,v')=2^{1/p}\geq 1$ $\forall v\neq v'$,
so $\dot\cV=\cV$ for $\eps=\fr12$.
Since also $d_\cV\leq 2=:r$, we also have $B_2(v)=\cV$.
Therefore $|\phi_s(B_2(v))|=|B_2(v)|=|\cV|=\infty>2^{sr+c\log s}$,
hence by \cref{lem:KraftIneq}, $\cV$ and therefore none of the spaces 
in the proposition scale-embed into $\BsdKs$.  
\end{proof}

That is, $d_K$ admits embeddings of (some) spaces 
of arbitrarily high finite dimension (\cref{thm:Ham2K}),
but not (some) infinite-dimensional embeddings.
$\BsdKs$ may be viewed as finite-but-unbounded dimensional,
similar to $\SetR^*:=\bigcup_{m=0}^\infty \SetR^m$.

Apart from limitations to finite-dimensional spaces,
$\BsdKs$ accommodating. 
Unfortunately the next result shatters any dream 
of $\BsdKs$ accommodating \emph{all} finite metrics.
\cref{prop:point5indK} does not extend to 6 points.
The brute-force proof may not be very satisfying,
but it does settle one of the most interesting questions.

\begin{proposition}[Not all finite metrics embed into $\BsdKs$]\hfil\par\label{prop:K33eindK}
\noindent $K_{3,3}$ with one edge removed does \emph{not} scale-embed into $\BsdKs$.
\end{proposition}

\begin{proof}
{\bf\boldmath Entropic vectors and linear programming:}
(See \cite{Ho:20} for background and references therein for this method)
Consider metric $d_M$ over $M=\{1,...,m\}$.
Assume $d$ iso-embeds into $d_I^+$ (recall \cref{def:distance}).
That means there are random variables $X_1,...,X_m$
with $d_M(i,j)=H(X_i|X_j)+H(X_j|X_i)$.
Let $X_\alpha$ with $\alpha\subseteq M:=\{1,...,6\}$ be a subset of these variables
(cf.\ the proof of \cref{thm:randomization}).
We know that for any $\alpha,\beta,\gamma\subseteq M$,
the conditional mutual information $I(X_\alpha;X_\beta|X_\gamma)\geq 0$
is non-negative, the so-called Shannon-inequalities.
They can be converted into inequalities only involving linear combinations of joint entropies $H(X_\delta)$
for various $\delta\subseteq M$.
Let $h=(h_\delta)_{\delta\subset M}\in\SetR^{2^M}$ be the vector of all such joint entropies $h_\delta:=H(X_\delta)$.
It is easy to see that a vector $h$ satisfies all Shannon-type inequalities if and only if $h$ is a 
monotone ($h_{\alpha\cup\beta}\geq h_\alpha$) and 
submodular ($h_{\alpha\cup\beta}+h_{\alpha\cap\beta}\leq h_\alpha+h_\beta$) function.
The set of all such $h$ is denoted by $\Gamma_m$ and characterized 
(being a bit careful) by the $2^{m-1}m$ monotonicity inequalities and 
$2^{m-3}(m-1)m$ Shannon-type inequalities.
Our equality $({m\atop 2})$ constraints are 
$h\in\Delta_G:=\{h:h_{\{i,j\}}-h_{\{j\}}+h_{\{i,j\}}-h_{\{j\}}=d_G(i,j)\}$.
Any iso-embedding of $d_M$ would lead to an $h\in\Gamma_m\cap\Delta_G$.
If we can show that $\Delta_G\cap \Gamma_m$ is empty, 
this implies that such an imbedding is impossible
(The converse would not be true due to the existence of non-Shannon-type inequalities).
This is a linear programming problem with $2^m$ variables, 
and $2^{m-3}(m+3)m$ constraints.

{\bf\boldmath $K_{3,3}\!\setminus\!e\not\to d_I$:}
For $G=K_{3,3}\!\setminus\!e$ ($m=6$) , 
the linear program with 64 variables and 432 constraints has no solution (the computer says).
For $d_I^\vee$ we eliminate the max by considering $2^{15}$ $=$ vs $\leq$ linear programming cases.
None of the cases leads to a solution.
We also checked that there is no solution if allowing for some $\eps$-slack (up to $0.05$)
to rule out numerical concerns and asymptotic embeddings.
Let $\Gamma_m^\eps$ be the such $\eps$-expanded polyhedral cone.
That is, we have $\Delta_G\cap \Gamma_m^\eps=\emptyset$, 
and therefore also $s\Delta_G\cap \Gamma_m^{s\eps}=\emptyset$, 
which we need below.

{\bf\boldmath $d_K\to d_I$:}
Assume $G$ scale-embeds into $\BsdKs$,
that is $s\cdot d_G(i,j)\equs d_K(x_i,x_j)$ for some $x_j$.
By \cref{thm:randomization} there exist random variables $X_i$
such that $K(x_i|x_j)\equs-\lb P[X_i|X_j]$. 
Taking the expectation on both sides and symmetrizing we get 
$d_K(x_i|x_j)\equs d_I(X_i,X_j)$.
Therefore $|s\cdot d_G(i,j)-d_I(X_i,X_j)|=O(\log s)$.
For $G=K_{3,3}\!\setminus\!e$, we showed $s\Delta_G\cap \Gamma_m^{s\eps}=\emptyset$,
which implies $|s\cdot d_G(i,j)-d_I(X_i,X_j)|\geq s\eps$, a contradiction.
\end{proof}

\section{\boldmath Euclidean Properties of $d_K$ and $d_G$}\label{sec:Euclid}

One way to show that an embedding is impossible is to show 
that $d_K$ has certain properties that $d_\cV$ lacks.
The next theorem shows that $d_K$ is \emph{not} Euclidean,
but this does not imply that Euclidean metrics cannot be embedded,
since there may be subsets of $\SetB^*$ on which $d_K$ is Euclidean.
To show this we need to prove some lemmas first.

\begin{lemma}[Euclidean Metric properties]\label{lem:metric}\hfil\par\ntc{\vspace{-2ex}}
\begin{itemize}\parskip=0ex\parsep=0ex\itemsep=0ex
\item[$(i)$] If $d$ is a (Euclidean) metric, then also $d^\alpha(x,y):=d(x,y)^\alpha$ for $\alpha\in[0;1]$.
\item[$(ii)$] If $d\geq 0$ is symmetric and CND, then also $d^\alpha$ for $\alpha\in[0;1]$. %
\item[$(iii)$] $d$ is an Euclidean metric \emph{iff} $d$ is a (semi)metric and $d^2$ is CND.
\item[$(iv)$] A metric on $\cX$ is Euclidean \emph{iff} it is Euclidean on every finite subset of $\cX$.
\item[$(v)$] Every separable Euclidean space can be isometrically embedded into $\ell_2$.
\item[$(vi)$] If $(\cV,d_\cV)$ and $(\cW,d_\cW)$ are Euclidean, then $(\cV\times\cW,\sqrt{d_\cV^2+d_\cW^2})$ is Euclidean.
\item[$(vii)$] $(\SetR^m,||\cdot||_p^{p/2})$ for $p\in[0;2]$ and esp.\ $(\SetR^m,\sqrt{||\cdot||_1})$ is Euclidean for $m\leq\infty$.
\end{itemize}
\end{lemma}

Item $(v)$ justifies calling spaces embeddable into Hilbert space `Euclidean', 
since in the separable case, one can always choose $\ell_2$ with the common Euclidean distance
and this even generalizes to non-separable spaces.
That $(i)$ $\sqrt{||\cdot||_2}$ and $(vii)$ even $\sqrt{||\cdot||_1}$  
are Euclidean (in this sense) is quite remarkable.
To give an intuition of the latter, 
map $v\in\{0,\fr1n,\fr2n,...,\fr{n-1}n,1\}^m$ to $x_v\in(\fr1n\SetB)^{m\cdot n}$
via $\phi(v)_i:=\fr1n(0^{v\cdot n}1^{n-v\cdot n})$.
Then $\sqrt{||v-v'||_1}=\sqrt{||x_v-x_{v'}||_1}=\sqrt{n}||x_v-x_{v'}||_2$, which is of course Euclidean.
By choosing $n$ large enough, we can hence \emph{approximately} embed $([0,1]^m,\sqrt{||\cdot||_1})$ into Euclidean space. 
Unfortunately we cannot take the limit $n\to\infty$, so $(vii)$ needs to be proven differently.

\begin{proof}{\bf(sketch)}
{\boldmath$(i)$} (P,Z,N,S) are obviously preserved under any power $\alpha\geq 0$.
(TI) is preserved, since $x^\alpha$ is concave for $\alpha\leq 1$, 
implying $(a+b)^\alpha\leq a^\alpha+b^\alpha$.
That the Euclidean property is preserved follows from $(ii)$ and $(iii)$. 
\\
{\boldmath$(ii)$} See \cite[Cor.3.2.10]{Berg:84}.
\\
{\boldmath$(iii)$} This is Schoenberg's embedding theorem \cite[Thm.9.6]{Paulsen:16}. 
The proof is instructive, so we present it here, but gloss over considerations related to (N). 
\cite[Ex.3.2.13b]{Berg:84}: If $d$ is Euclidean,
then $d^2(x,y)=||x-y||_2^2=\langle x,x\rangle+\langle y,y\rangle -2\langle x,y\rangle$ 
for some inner product $\langle\cdot,\cdot\rangle$.
Inner products are PD, hence $-\langle x,y\rangle$ is CND.
Any CND remains CND if adding any function that 
depends only on $x$ or $y$, hence $d^2$ is CND too.
For the other direction \cite[Prop.3.3.2]{Berg:84}: 
Assume $d$ is a (semi)metric on $\cX$ and $d^2$ is CND, and choose any $x_0\in\cX$.
Then it is easy to see that $\phi(x,y):=\fr12[d(x,x_0)^2+d(y,x_0)^2-d(x,y)^2]$ is PD \cite[Lem.3.2.1]{Berg:84}
(Plug $c_0:=-c_1-...-c_m$ into $\sum_{i,j=0}^m c_i c_j d(x_i,x_j)^2\leq 0$).
Therefore, $\phi$ can be represented as an inner product 
$\phi(x,y)=\langle\varphi_x,\varphi_y\rangle$ with vector $\varphi_x:=\phi(x,\cdot)$.
Elementary calculus now shows $||\varphi_x-\varphi_y||_2^2=\langle \varphi_x,\varphi_x\rangle+\langle \varphi_y,\varphi_y\rangle -2\langle \varphi_x,\varphi_y\rangle = ...= d(x,y)^2$.
Therefore $d$ is Euclidean. 
\\
{\boldmath$(iv)$}
The only if direction is trivial. The if direction follows from $(iii)$:
If $d$ is (E) on every finite subset, then $d^2$ is CND on every finite subset,
therefore it is CND on all of $\cX$ (by the mere definition of CND), 
therefore $d$ is (E) on $\cX$.
\\
{\boldmath$(v)$} By definition, if $(\cV,d)$ is Euclidean, it can be embedded into a Hilbert space $H$.
Since $\cV$ is assumed separable, and all isometric embeddings are continuous (indeed 1-Lipschitz),
$H$ is separable too, and every separable Hilbert space is isometrically isomorphic to $\ell_2$.
\\
{\boldmath$(vi)$} By $(iii)$, $d_\cV^2$ and $d_\cW^2$ are CND, hence $d_\cV^2$+$d_\cW^2$ is CND,
so $\sqrt{d_\cV^2+d_\cW^2}$ is Euclidean, again by $(iii)$.
\\
{\boldmath$(vii)$} $(\SetR,|\cdot|$) is Euclidean, so by $(i)$, 
also $(\cV,d_\cV):=(\SetR,|\cdot|^{p/2})$ and $(\cW,d_\cW):=(\SetR,|\cdot|^{p/2})$ for $p\leq 2$. 
Now by $(vi)$, also $\sqrt{d_\cV^2(v,v')+d_\cW^2(w,w')}=\sqrt{|v-v'|^p+|w-w'|^p}=\sqrt{||u-u'||_p^p}$,
is Euclidean, where $u:=(v,w)$ and $u':=(v',w')$. Therefore $(\SetR^2,\sqrt{||\cdot||_p^p})$ is Euclidean.
Now repeat for finite $m$ and take the limit for $m=\infty$.
\end{proof}

\begin{theorem}[$d_K^\alpha$ is not Euclidean for $\alpha>0.3$]\label{thm:dKnotE} 
\end{theorem}

\begin{proof}
The last case ($\alpha>0.3$) would already imply the result,
but it is instructive to first prove the easier cases $\alpha=1$ and $\alpha>\fr12$ separately.
\\
{\boldmath\bf $\alpha>1$:} In this case, $d_K$ is not even a metric:
Consider the $K$-random and $K$-independent strings $x_{1:2s}$ and $z_{1:2s}$,
and dependent $y_{1:2s}:=x_{1:s}z_{1:s}$, then 
$$
  K(x_{1:2s}|z_{1:2s}) ~\equs~ 2s ~=~ s+s ~\equs~ K(x_{1:2s}|y_{1:2s})+K(y_{1:2s}|z_{1:2s})
$$
and similarly for $x\leftrightarrow z$, hence this also holds for $d_K$.
If $c=a+b$ then $c^\alpha>a^\alpha+b^\alpha$ for $\alpha>1$ if $a,b>0$.
Similarly if $c_s\equs a_s+b_s$, 
then $c_s^\alpha>a_s^\alpha+b_s^\alpha$ for sufficiently large $s$,
hence $d_K^\alpha$ violates (TI) on these 3 strings.
\\
{\boldmath\bf $\alpha=1$:} 
Consider the set of 4 strings $\cX=\{000,~001,~010,~100\}$ with Hamming distance $d_H$.
The resulting distance matrix is
$$
  d_H ~=~ \small\arraycolsep=3pt\left(\begin{array}{c|ccc}
    0 & 1 & 1 & 1 \\ \hline
    1 & 0 & 2 & 2 \\
    1 & 2 & 0 & 2 \\
    1 & 2 & 2 & 0 \\
  \end{array}\right)
$$
Consider the element-wise square $M_{ij}:=(d_H)_{ij}^2$.
For $c=(3,-1,-1,-1)$ (which sums to $0$) we get $c^\top M c=6>0$,
so $M$ is not CND. 
\cref{lem:metric}$(iii)$ implies that $d_H$ is not Euclidean.
We can scale-embed $d_H$ into $\BsdKs$ via \cref{thm:Ham2K}.
Hence $d_K$ is also not Euclidean (not even if we allowed $O(1)$ slack).
\\
{\boldmath\bf $\alpha>\fr12$:}
We can generalize the above proof from $m:=|\cX|$ ($=4$ above) to larger 
$\cX=\{0^{m-1},~0^{m-2}1,...,~10^{m-2}\}$ in the obvious way.
Then metric $(d_H)_{1j}=1$ for $j>1$ and $(d_H)_{ij}=2$ for $1<i\neq j>1$,
and of course $(d_H)_{ii}=0~\forall i$.
Consider the element-wise power $M_{ij}:=(d_H)_{ij}^{2\alpha}$.
For $c=(m-1,-1,...,-1)\in\SetR^m$ (which sums to $0$) we get 
$$
  c^\top\! M c ~=~...~=~ (m-1)[(m-2)2^{2\alpha}-2(m-1)] 
  ~>~ 0 \ntc{~~~\text{iff}~~~ 2^{2\alpha-1} ~>~ \fr{m-1}{m-2}}
$$
\tc{iff $2^{2\alpha-1}>\fr{m-1}{m-2}$.}
For any $\alpha>\fr12$ we can choose $m$ large enough 
such that the condition is satisfied,
so $M$ is not CND in this case, 
hence $d_H^\alpha$ is not Euclidean by \cref{lem:metric}$(iii)$,
hence $d_K^\alpha$ is not Euclidean by \cref{thm:Ham2K}.
\\
{\boldmath\bf $\alpha>0.3$:}
By \cref{thm:dGnotE}, $d_G^\alpha$ for the path metric $d_G$ 
of the complete bipartite graph $G=K_{3,3}$ 
is not Euclidean for $4^\alpha>\fr32$,
but by the proof of \cref{prop:K32indK} $(G,d_G)$ scale-embeds into $\BsdKs$,
hence $(G,d_G^\alpha)$ scale-embeds into $(\SetB^*,d_K^\alpha)$,
hence $d_K^\alpha$ cannot be Euclidean either for $\alpha>\log_4\fr32=0.292...$.
\end{proof}

While $d_K$ is not Euclidean on all of $\SetB^*$, 
it can accomodate any finite Euclidean point set:

\begin{theorem}[Every finite Euclidean point set sequence-embeds into $d_K$]\label{prop:finEtoDK}
\end{theorem}

\begin{proof}
By \cite[Ex.15.5.5b]{Matousek:02},
$\ell_2$ isometrically embeds into $\ell_1$.
By \cite[Ex.15.5.2c]{Matousek:02} every finite subset $\cV$ of $\ell_1$ of size $|\cV|=m$
iso-embeds into $(\SetR^{(\!{m\atop 2}\!)},||\cdot||_1)$,
and indeed into $[a,b]^{(\!{m\atop 2}\!)}$ for some $-\infty<a<b<\infty$, since $\cV$ is finite.
Finally $[a,b]^{(\!{m\atop 2}\!)}$ sequence-embeds into $\BsdKs$ by \cref{cor:OneNormEinK}.
\end{proof}

Despite most AIT results suffering from additive slacks,
we were able to define $d_K$ to be an exact metric without any slack.
We can also ask whether $d_K$ is exactly Euclidean on some subspace of $\SetB^*$.

\begin{proposition}[$d_K$ is exactly Euclidean on some finite subsets of $\SetB^*$]\label{prop:dKEsubset}
There are arbitrarily large finite subsets of $\SetB^*$ on which $d_K$ is Euclidean, 
and the points in these sets are mapped to linearly independent vectors.
\end{proposition}

Without the independence qualifier, the proposition would already follow from \cref{thm:Interval2K},
since $\SetR^1$ is an Euclidean space ($|\cdot|=||\cdot||_2$), and $[0,1]$ is infinite.
Linear independence means that the points span a space of arbitrarily large finite dimension.

\begin{proof}
Let $x^1,...,x^m\in\SetB^\infty$ be 
$m$ (pairwise) $K$-independent $K$-random binary sequences,
i.e.\ $K(x^i_{1:s}|x^j_{1:s}) \equs s$ for $i\neq j$.
Hence $d_K(x^i_{1:s},x^j_{1:s})/s\to[\![i\neq j]\!]=d_{01}(i,j)$.
Let $\Delta=\{e^1,...,e^m\}$ be the corners of a standard $m-1$-simplex in $\SetR^m$, 
where $e^i$ is the $i$th unit vector. Then $||e^i-e^j||_2/\sqrt{2}=d_{01}(i,j)$, which is Euclidean.
Hence $\phi_s(e^i)=x_{1:s}^i$ sequence-embeds $(\Delta,||\cdot||_2/\sqrt{2})$ into $(\SetB^\infty,d_K)$,
and $\Delta$ spans an $m-1$-dimensional sub-space of $\SetR^m$.
This proves $\lim_{s\to\infty}d_K/s$ is Euclidean.
For sufficiently large $s$, $d_K/s$ is $\eps$-close to $d_{01}$.
For small-enough $\eps$, $d_{01}$ embeds into a slightly distorted standard simplex (Exercise).
Therefore, $d_K$ is exactly Euclidean on $\{x_{1:s}^1,...,x_{1:s}^m\}$ for all sufficienly large $s$.
\end{proof}

For any finite metric space $(\cX,d)$, 
$d^\alpha$ is Euclidean for sufficiently small $\alpha>0$.
This follows simply from the fact that 
$d(x,y)^\alpha\to[\![x\neq y]\!]=d_{01}(x,y)$ for $\alpha\to 0$, which is Euclidean.
But the larger $\cX$ the smaller $\alpha$ may need to be, as the next theorem shows,
which is used in the proof of \cref{thm:dKnotE} for $\alpha\leq\fr12$.

\begin{theorem}[Graphs for which $d_G^\alpha$ is not Euclidean for $\alpha>0$]\label{thm:dGnotE} 
Consider the complete (unweighted) bipartite graph $G=K_{m,n}$ of $m+n$ nodes ($m,n\geq 2$)
and its induced metric $d_G$.
Then for any $0<\alpha\leq 1$, there exist $m,n$ such that $d_G^\alpha$ is \emph{not} Euclidean.
More precisely, $d_G^\alpha$ is Euclidean iff $4^\alpha\leq 2/(2-\fr1m-\fr1n)$.
For $G=K_{\infty,\infty}$, $d_G^\alpha$ is not Euclidean for all $\alpha>0$. 
\end{theorem}

\begin{proof}
For $G=K_{m,n}$ we have 
$(d_G)_{ij}=2$ if $i\neq j$ are both $1\leq i,j\leq m$ or both $m<i,j\leq m+n$,
and $1$ otherwise, and of course $(d_G)_{ii}=0~\forall i$.
For instance, for $m=3$ and $n=2$,
$$
  d_G ~=~ \scriptsize\arraycolsep=2pt \left(\begin{array}{ccc|cc}
    0 & 2 & 2 & 1 & 1 \\
    2 & 0 & 2 & 1 & 1 \\
    2 & 2 & 0 & 1 & 1 \\ \hline
    1 & 1 & 1 & 0 & 2 \\
    1 & 1 & 1 & 2 & 0 \\
  \end{array}\right)
$$
Consider the element-wise power $M_{ij}:=(d_G)_{ij}^{2\alpha}$.
For $c=(\fr1n,...,\fr1n,\fr{-1}m,...,\fr{-1}m)\in\SetR^{m+n}$ 
($m\times$ $\fr1m$ and $n\times$ $\fr{-1}n$, which sums to $0$) we get
$$
  c^\top\! M c ~=~ (2-\fr1m - \fr1n)4^\alpha-2 ~>~ 0 
  ~~~\text{iff}~~~ 4^\alpha ~>~ \fr{2}{2-1/m-1/n}
$$
For any $\alpha>0$ we can choose $m,n$ large enough such that the condition is satisfied,
so $M$ is not CND in this case. 
By \cref{lem:metric}$(iii)$ this implies that $d_G^\alpha$ is not an Euclidean metric iff
$4^\alpha>\fr{2}{2-1/m-1/n}$.
The converse is also true by noting that among all $c$
with $\sum_i c_i=0$, the above $c$, maximizes $c^\top\!M c$ (Exercise).
For $G=K_{\infty,\infty}$ and any given $\alpha>0$, choose suitably large $m,n$.
Since $K_{m,n}$ is a subgraph of $K_{\infty,\infty}$,
$d_G^\alpha$ is not Euclidean.
\end{proof}

In particular, for $K_{3,2}$ and $\alpha=\fr12$, $4^\alpha=2\not\leq 12/7=2/(2-\fr13-\fr12)$, 
hence $\sqrt{d_{K_{3,2}}}$ is not Euclidean, but $(\SetB^{m'},\sqrt{d_H})$ is Euclidean.
Therefore $K_{3,2}$ cannot be embedded into the Hamming cube.
This prevents the Hamming route via \cref{thm:Ham2K} to strengthen \cref{thm:dKnotE} to $\alpha\leq\fr12$.
Furthermore, by \cref{lem:metric}$(vii)$, $\sqrt{d_{K_{3,2}}}$ 
does not embed into $(\SetR^m,\sqrt{||\cdot||_1})$ either, 
hence $K_{3,2}$ does not embed into $(\SetR^m,||\cdot||_1)$,
therefore the route via \cref{cor:OneNormEinK} is blocked too.
While these indirect embedding strategies fail,
it actually \emph{is} possible to scale-embed $K_{3,2}$ directly into $\BsdKs$.

This still does not settle whether many interesting metrics 
that do not embed into $\ell_p$ could nevertheless embed into $\BsdKs$.
In particular:

\begin{open}[Embedding bipartite graphs]\label{open:dKbipartite}
Can the bipartite graphs $K_{m,n}$ be scale-embedded into $\BsdKs$ for all $m,n$?
\end{open}

A positive solution would, via \cref{thm:dGnotE}, imply 
that $d_K^\alpha$ is not Euclidean for any $\alpha>0$ (settling \cref{open:dKEuclid}) not just $\alpha>0.3$.
On the other hand, the fact that $K_{3,2}$ iso-embeds into $\BsdKs$ but not $\ell_2$,
prevents iso-embedding $\BsdKs$ into $\ell_p$ also for $p<2$.

\begin{theorem}[$d_K$ does not iso-embed into $\ell_p$ for $1\leq p\leq 2$]\label{thm:dHnotinl1}
\end{theorem}

\begin{proof}{\bf(sketch)}
The result follows from the following diagram:
$$
\begin{array}{ccccc}
          & \rotatebox[origin=tr]{30}{$\stackrel{s}{\longrightarrow}$} &     d_K      & \hspace{4pt}/\hspace{-10pt}\smash{\stackrel{?!}\longrightarrow} & \ell_p ~(1\leq p\leq 2)\hspace{-10ex} \\
  K_{3,2} &          & ^?\downarrow\!\!\!\!^{-!} &             & \downarrow          \\
          & \smash{\rotatebox[origin=br]{-30}{$\hspace{4pt}|\hspace{-10pt}\longrightarrow$}} & \ell_2     & \longrightarrow & \ell_1
\end{array}
$$
\Cref{thm:dGnotE} shows that the complete bipartite graph $K_{3,2}$ is not Euclidean 
but scale-embeds by \Cref{prop:K32indK} into $\BsdKs$,
hence $d_K$ cannot be Euclidean without causing a contradiction.
A deep and remarkable theorem in metric spaces \cite{Bretagnolle:66} 
shows that $\ell_p$ iso-embeds into $\ell_1$ for $1\leq p\leq 2$,
while $K_{3,2}$ does not scale-embed into $\ell_1$ 
by the comment after the proof of \cref{thm:dGnotE},
hence $d_K$ cannot iso-embed into $\ell_p$ either.
\end{proof}

Given \cref{thm:dKnotE,thm:dHnotinl1}, we can lower our expectations and ask:

\begin{open}[Is $(\SetB^*,d_K^\alpha)$ Euclidean for some $\alpha>0$?]\label{open:dKEuclid}
\end{open}

A weak argument for its truth might be that $\sqrt{d_H}$ is Euclidean and 
embeds into $\sqrt{d_K}$, and both are linearly additive in nature.
A stronger argument against is that \cref{thm:dKnotE} already excludes $\sqrt{d_K}$ 
and the additive nature only holds
if parts of the strings are $K$-independent, 
which is true for random strings but not true in general.
Also, given its generality and unusual nature, 
one would expect $d_K$ to satisfy very few ``classical'' properties.
One precedence of an information distance that is not a metric, 
but its square-root is a metric, is the Jensen--Shannon divergence $d_{JS}$,
albeit between distributions, not between random variables.

\section{Conclusion}\label{sec:Disc}

We investigated properties of the complexity distance $d_K$
over the space of binary strings $\SetB^*$,
defined as the symmetrized conditional Kolmogorov complexity $K(x|y)$
(with some additive correction to make it a proper metric).
We did so by asking which other metric spaces can be scale-embedded into $\BsdKs$.
For instance, we showed that $([0,1]^m,||\cdot||_1)$ does so for $m<\infty$ but does not for $m=\infty$.
We also showed that $d_K$ is not Euclidean,
but many questions have been left open, e.g.\ whether $d_K$ raised to some small power is Euclidean.
We do not even know whether all finite metric spaces can be scale-embedded into $\BsdKs$.

Apart from solving the stated 
\cref{open:dKuniversal,open:dKauto,open:PSEinK,open:dKgrid,open:dKbipartite,open:dKEuclid}
in the text, especially the grand \cref{open:dKuniversal},
there are many related questions we have not touched:
Natural is to consider the \emph{Normalized} version of the Information/Compression Distance (NID/NCD),
which enjoyed great practical success \cite{Cilibrasi:05}.
We could ask whether $\BsdKs$ possesses any of the other properties some classical metrics possess,
for instance whether it is convex or geodesic, 
provided these notions can be suitably adapted to make sense in $\BsdKs$.
Since every metric induces a topology, we could also ask about the weaker topological properties.
We could also relax the conditions for and hence easing scale-isometries: The weakest condition
that keeps the sprit of $\phi_s$ being asymptotically exact would be to only require that
$d_K(\phi_s(v),\phi_s(v'))/s$ converges to $d(v,v')$ for $s\to\infty$,
or any notion between these two extremes.
If convergence is non-uniform in $v,v'$, this would imply that for each fixed $s$,
there are $v,v'$ for which the accuracy is arbitrarily bad, so this notion may be of limited interest.
There is a close connection between Shannon entropy and Kolmogorov complexity \cite[Tab.2.22]{Hutter:24uaibook2},
for instance, all linear inequalities that hold for one also hold for the other \cite{Hammer:00},
so it is natural to explore whether this connection can help here as well.
Even more so, under mild conditions, Shannon entropy converges to Kolmogorov
complexity \cite{Austern:20}, hence $d_I=d_K\pm O(1/\sqrt{s})$.
We could weaken the requirement of asymptotic scale-embeddings further to quasi-isometries,
or even further to approximate isometries with distortions growing with the number of points
like in the famous Johnson--Lindenstrauss flattening lemma \cite[Chp.15]{Matousek:02}.


\section*{References}\label{sec:Bib}
\addcontentsline{toc}{section}{\refname}

\def\refname{\vspace{-4ex}}
\begin{small}

\end{small}

\onecolumn\clearpage\appendix

\section{List of Notation}\label{app:Notation}

\par\vspace{0pt plus \textheight}
\begin{samepage} 
\begin{tabbing} 
  \hspace{0.13\textwidth} \= \hspace{0.73\textwidth} \= \kill
  {\bf Symbol }     \> {\bf Explanation}                                                    \\
  $a/bc=a/(bc)$     \> ~ while $a/b\cdot c=(a/b)\cdot c$ though we actually always bracket the latter \\
  $m$               \> dimension or size of set. By default $\in\SetN$                      \\
  $i,j\in\SetN$     \> vector or matrix indices $\in\{1,...,m\}$                            \\
  $n$               \> generic natural number                                               \\
  $k,d$             \> kernels, distances $\cX\times\cX\to\SetR$                            \\
  $(\cX,d_\cX)$     \> metric space. Set $\cX$ with metric $d_\cX$                          \\
  $K,\Km$           \> prefix Kolmogorov complexity, monotone complexity                    \\
  $\SetB$           \> $=\{0,1\}$                                                           \\
  $[\![\text{\it bool}]\!]$ \> =1 if {\it bool}=True, =0 if {\it bool}=False (Iverson bracket) \\
  $|\cX|$           \> size of set $\cX$.                                                   \\
  $\cA$             \> finite alphabet like $\{a,...,z\}$ or ASCII or $\SetB=\{0,1\}$.      \\
  $\cA^n,\cA^*$     \> set of all (length $n$, finite) strings over alphabet $\cA$          \\
  $\cA^\infty$      \> set of all infinite sequences over alphabet $\cA$                    \\
  $x,y,z\in\SetB^*$  \> finite binary strings in metric space $\BsdKx$               \\
  $x,y,z\in\SetB^\infty$ \> infinite binary sequences, initials used in $\BsdKx$     \\
  $\BsdKs$          \> $\SetB^*$ equipped with metric $d_K$.                                \\
  $\ell(x)$         \> length of string $x$                                                 \\
  $x_{1:s}\in\cA^s$ \> string of length $s$                                                 \\
  $x_{<t}\in\cA^{t-1}$\> string of length $t-1$                                             \\
  iff               \> if and only if                                                       \\  
  w.l.g.            \> without loss of generality                                           \\
  w.p. ...          \> with probability ...                                                 \\
  $\eps,\epstr$     \> small number $>0$, empty string                                      \\
  $\lb$             \> binary (base-$2$) logarithm $\log_2$                                 \\
  $\SetR,\SetN,...$ \> set of real,natural numbers                                          \\
  $\ell_p$          \> $:=\{v\in\SetR^\infty:||v||_p<\infty\}$ for some $p\leq\infty$       \\
  $O()$             \> classical $O()$ notation                                             \\
  $\leqa,\geqa,\equa$ \> $\leq,\geq,=$ within a universal additive constant $c_U$           \\
  $\leqs,\geqs,\equs$ \> $\leq,\geq,=$ within $\pm O(\log s)$                               \\
  $:=,~\equiv$      \> definition, equal by earlier definition                              \\
  $u,v,w$           \> elements of metric spaces ($\cV,d_\cV)$ etc.                         \\
  $v',x',v'',x''$   \> generic elements of the same set as $v$, $x$, etc.                   \\
  $\dot v, \dot x$  \> element of $\dot\cV\subseteq\cV$, $\dot\cX\subseteq\cX$, etc.        \\
  $\cU,\cV,\cW$     \> metric spaces with metric $d_\cU,d_\cV,d_\cW$                        \\
  $s\in\SetN$       \> metric scale factor                                                  \\
  $\phi_s$          \> scale-embedding, (e.g.\ $x_u^s:=\phi_s(u)$ from $(\cU,d_\cU)$) into $\BsdKs$ \\
  $\phi_*$          \> scale-embedding (e.g.\ from $(\cV,d_\cV)$) into $((\SetB^*)^\infty,d_K)$ \\
  $\phi$            \> uniform scale-embeddings (e.g. from $(\cW,d_\cW)$) to $(\SetB^\infty,d_K)$ \\
  $P(x|y)$          \> short for $P[X=x|Y=y]$. Probability that $X=x$ given $Y=y$           \\
  $P[X|Y]$          \> random variable $f(X,Y)$, where $f(x,y):=P(x|y)$                     
\end{tabbing}
\end{samepage}


\begin{thebibliography}{ABCD}\parskip=0ex
\newcommand{\etalchar}[1]{$^{#1}$}

\bibitem[AM20]{Austern:20}
Morgane Austern and Arian Maleki.
\newblock On the {{Gaussianity}} of {{Kolmogorov Complexity}} of {{Mixing Sequences}}.
\newblock {\em IEEE Transactions on Information Theory}, 66(2):1232--1247, February 2020.

\bibitem[BCK66]{Bretagnolle:66}
Jean Bretagnolle, Didier~Dacunha Castelle, and Jean-Louis Krivine.
\newblock Lois stables et espaces $l^p$.
\newblock {\em Annales de l'I.H.P. Probabilités et statistiques}, 2(3):231--259, 1966.

\bibitem[BCR84]{Berg:84}
Christian Berg, Jens Peter~Reus Christensen, and Paul Ressel.
\newblock {\em Harmonic Analysis on Semigroups: Theory of Positive Definite and Related Functions}.
\newblock Number 100 in Graduate Texts in Mathematics. Springer, New York Berlin Heidelberg [etc.], 1984.

\bibitem[BGM{\etalchar{+}}98]{Bennett:98}
C.H. Bennett, P.~Gacs, {Ming Li}, P.M.B. Vitanyi, and W.H. Zurek.
\newblock Information distance.
\newblock {\em IEEE Transactions on Information Theory}, 44(4):1407--1423, July 1998.

\bibitem[CHS07]{Hutter:07postbndx}
Alexey Chernov, Marcus Hutter, and J{\"u}rgen Schmidhuber.
\newblock Algorithmic complexity bounds on future prediction errors.
\newblock {\em Information and Computation}, 205(2):242--261, 2007.

\bibitem[CV05]{Cilibrasi:05}
R.~Cilibrasi and P.~M.~B. Vit{\'a}nyi.
\newblock Clustering by compression.
\newblock {\em IEEE Trans. Information Theory}, 51(4):1523--1545, 2005.

\bibitem[CV22]{Cilibrasi:22}
Rudi~L. Cilibrasi and Paul M.~B. Vit{\'a}nyi.
\newblock Fast {{Phylogeny}} of {{SARS-CoV-2}} by {{Compression}}.
\newblock {\em Entropy}, 24(4):439, March 2022.

\bibitem[Fre10]{Frechet:10}
M.~Frechet.
\newblock Les dimensions d'un ensemble abstrait.
\newblock {\em Mathematische Annalen}, 68:145--168, 1910.

\bibitem[FT04]{Fuglede:04}
B.~Fuglede and F.~Topsoe.
\newblock Jensen-{{Shannon}} divergence and {{Hilbert}} space embedding.
\newblock In {\em International {{Symposium onInformation Theory}}, 2004. {{ISIT}} 2004. {{Proceedings}}.}, pages 30--30, Chicago, Illinois, USA, 2004. IEEE.

\bibitem[Gri92]{Grishukhin:92}
V.P. Grishukhin.
\newblock Computing extreme rays of the metric cone for seven points.
\newblock {\em European Journal of Combinatorics}, 13(3):153--165, May 1992.

\bibitem[GW85]{Graham:85}
R.~L. Graham and P.~M. Winkler.
\newblock On isometric embeddings of graphs.
\newblock {\em Transactions of the American Mathematical Society}, 288(2):527--536, February 1985.

\bibitem[HHO25]{Hutter:25kccluster}
Boumediene Hamzi, Marcus Hutter, and Houman Owhadi.
\newblock Bridging algorithmic information theory and machine learning: {{Clustering}}, density estimation, {{Kolmogorov}} complexity-based kernels, and kernel learning in unsupervised learning.
\newblock {\em Physica D: Nonlinear Phenomena}, 476:134669, 2025.

\bibitem[HLTY20]{Ho:20}
Siu-Wai Ho, Lin Ling, Chee~Wei Tan, and Raymond~W. Yeung.
\newblock Proving and {{Disproving Information Inequalities}}: {{Theory}} and {{Scalable Algorithms}}.
\newblock {\em IEEE Transactions on Information Theory}, 66(9):5522--5536, September 2020.

\bibitem[HQC24]{Hutter:24uaibook2}
M.~Hutter, D.~Quarel, and E.~Catt.
\newblock {\em An Introduction to Universal Artificial Intelligence}.
\newblock Chapman \& {{Hall}}/{{CRC}} Artificial Intelligence and Robotics Series. {Taylor and Francis}, 2024.

\bibitem[HRSV00]{Hammer:00}
Daniel Hammer, Andrei Romashchenko, Alexander Shen, and Nikolai Vereshchagin.
\newblock Inequalities for {{Shannon Entropy}} and {{Kolmogorov Complexity}}.
\newblock {\em Journal of Computer and System Sciences}, 60(2):442--464, April 2000.

\bibitem[Hus08]{Husek:08}
Miroslav Husek.
\newblock Urysohn universal space, its development and {{Hausdorff}}'s approach.
\newblock {\em Topology and its Applications}, 155(14):1493--1501, August 2008.

\bibitem[Hut05]{Hutter:04uaibook}
Marcus Hutter.
\newblock {\em Universal Artificial Intelligence: Sequential Decisions Based on Algorithmic Probability}.
\newblock Springer, Berlin, 2005.

\bibitem[JWBL23]{Jiang:23fewshot}
Zhiying Jiang, Rui Wang, Dongbo Bu, and Ming Li.
\newblock A {{Theory}} of {{Human-Like Few-Shot Learning}}, January 2023.

\bibitem[JYT{\etalchar{+}}23]{Jiang:23text}
Zhiying Jiang, Matthew Yang, Mikhail Tsirlin, Raphael Tang, Yiqin Dai, and Jimmy Lin.
\newblock ``{{Low-Resource}}'' {{Text Classification}}: {{A Parameter-Free Classification Method}} with {{Compressors}}.
\newblock In {\em Findings of the {{Association}} for {{Computational Linguistics}}: {{ACL}} 2023}, pages 6810--6828, Toronto, Canada, July 2023. Association for Computational Linguistics.

\bibitem[KN74]{Kuipers:74}
Lauwerens Kuipers and Harald Niederreiter.
\newblock {\em Uniform Distribution of Sequences}.
\newblock Pure and Applied Mathematics. Wiley, New York, 1974.

\bibitem[LCL{\etalchar{+}}04]{Li:04}
M.~Li, X.~Chen, X.~Li, B.~Ma, and P.M.B. Vitanyi.
\newblock The {{Similarity Metric}}.
\newblock {\em IEEE Transactions on Information Theory}, 50(12):3250--3264, December 2004.

\bibitem[LV07]{Li:07}
Ming Li and Paul M.~B. Vitanyi.
\newblock Applications of algorithmic information theory.
\newblock {\em Scholarpedia}, 2(5):2658, May 2007.

\bibitem[LV19]{Li:19}
Ming Li and Paul Vitanyi.
\newblock {\em An Introduction to {{Kolmogorov}} Complexity and Its Applications}.
\newblock Texts in Computer Science. Springer Berlin Heidelberg, New York, NY, 4th edition, 2019.

\bibitem[Mat02]{Matousek:02}
Jiri Matousek.
\newblock {\em Lectures on Discrete Geometry}.
\newblock Number 212 in Graduate Texts in Mathematics. Springer, New York, 2002.

\bibitem[PR16]{Paulsen:16}
Vern~I. Paulsen and Mrinal Raghupathi.
\newblock {\em An {{Introduction}} to the {{Theory}} of {{Reproducing Kernel Hilbert Spaces}}}.
\newblock Cambridge University Press, April 2016.

\bibitem[Sch12]{Schmidhuber:12}
J{\"u}rgen Schmidhuber.
\newblock A {{Formal Theory}} of {{Creativity}} to {{Model}} the {{Creation}} of {{Art}}.
\newblock In Jon McCormack and Mark {d'Inverno}, editors, {\em Computers and {{Creativity}}}, pages 323--337. Springer Berlin Heidelberg, Berlin, Heidelberg, 2012.

\bibitem[SUV17]{Shen:17}
A.~Shen, V.~A. Uspenski{\u \i}, and Nikolai~Konstantinovich Vereshchagin.
\newblock {\em Kolmogorov Complexity and Algorithmic Randomness}.
\newblock Number volume 220 in Mathematical Surveys and Monographs. American Mathematical Society, Providence, Rhode Island, 2017.

\bibitem[Ury27]{Urysohn:27}
Pavel~Samuilovich Urysohn.
\newblock Sur les espace m{\'e}trique universel.
\newblock pages 43--64 and 74--90, 1927.

\bibitem[VH08]{VanderMaaten:08}
Laurens {Van der Maaten} and Geoffrey Hinton.
\newblock Visualizing data using t-{{SNE}}.
\newblock {\em Journal of machine learning research}, 9(11), 2008.

\bibitem[{Wik}24]{Wikipedia:24kart}
{Wikipedia}.
\newblock Low-complexity art.
\newblock {\em Wikipedia}, December 2024.

\bibitem[ZXYL20]{Zeng:20}
Zhenbing Zeng, Yaochen Xu, Zhengfeng Yang, and Zhi-bin Li.
\newblock An isometric embedding of the impossible triangle into the euclidean space of lowest dimension.
\newblock In {\em Maple Conference}, pages 438--457, 2020.

\end{thebibliography}
\end{document}